\documentclass[journal,twoside,web]{ieeecolor}

\usepackage{generic}
\usepackage{xargs}
\usepackage{xparse}
\usepackage{amsmath,amssymb,amsfonts}
\usepackage{bm,bbm}
\usepackage{mathtools}
\usepackage{nicefrac}
\usepackage{graphicx}
\usepackage{xcolor}
\usepackage{pgf}
\usepackage{tikz}
\usepackage{booktabs}
\usepackage{multirow}
\usepackage{algorithm,algorithmic}
\usepackage{subcaption}
\usepackage{hyperref}
\hypersetup{hidelinks,colorlinks,urlcolor=black}
\usepackage{orcidlink}
\usepackage[backend=biber,style=ieee,citestyle=numeric-comp,sorting=none]{biblatex}
\newtheorem{theorem}{Theorem}

\newcommand{\N}{\mathbb{N}} 
\newcommand{\R}{\mathbb{R}} 

\newcommand{\vct}[1]{\bm{#1}} 
\newcommand{\mtx}[1]{\bm{#1}} 
\newcommand{\tns}[1]{\bm{\mathcal{#1}}} 

\newcommandx{\seq}[3][2=k\in\N,3={}]{(#1)_{#2}^{#3}}

\DeclareMathOperator*{\minimize}{minimize} 
\DeclareMathOperator*{\argmin}{arg\,min} 

\renewcommand{\toprule}{\specialrule{0.1em}{0em}{0em}}
\renewcommand{\midrule}{\specialrule{0.05em}{0em}{0em}}
\renewcommand{\bottomrule}{\specialrule{0.1em}{0em}{0em}}

\newcommand{\remove}[1]{}

\NewDocumentCommand{\includegraphix}{O{\textwidth} m g g g g g}{%
    \begin{tikzpicture}
        \node[anchor=south west] (image) at (0,0) {\includegraphics[width=#1]{#2}};
        \begin{scope}[shift={(image.south west)}]
            \IfValueT{#3}{%
                \node[anchor=south west, text=white, fill=black] at (4pt,4pt) {#3};
            }
            \IfValueT{#4}{%
                \IfValueT{#5}{%
                    \draw[orange, very thick] #4 circle (#5);
                }
            }
            \IfValueT{#6}{%
                \IfValueT{#7}{%
                    \draw[orange, very thick] #6 circle (#7);
                }
            }
        \end{scope}
    \end{tikzpicture}%
}

\def\BibTeX{{\rm B\kern-.05em{\sc i\kern-.025em b}\kern-.08em
    T\kern-.1667em\lower.7ex\hbox{E}\kern-.125emX}}
\begin{document}

\title{Deconver: A Deconvolutional Network for Medical Image Segmentation}

\author{
    Pooya Ashtari\,\orcidlink{0000-0002-5032-7808}, Shahryar Noei\,\orcidlink{0009-0006-1708-5591}, Fateme Nateghi Haredasht\,\orcidlink{0000-0002-8874-8835},
    Jonathan H. Chen\,\orcidlink{0000-0002-4387-8740}, Giuseppe Jurman\,\orcidlink{0000-0002-2705-5728},
    Aleksandra Pi\v{z}urica\,\orcidlink{0000-0002-9322-4999}, \IEEEmembership{Senior Member, IEEE}, and
    Sabine Van Huffel\,\orcidlink{0000-0001-5939-0996}
    \thanks{The research was partially funded by the Flanders AI Research Program and the National Plan for Complementary Investments to the NRRP (D34H project, code: PNC0000001).}
    \thanks{Pooya Ashtari is with the Department of Electrical Engineering (ESAT), STADIUS Center, KU Leuven, Belgium, and with the Department of Telecommunications and Information Processing, Ghent University, B-9000 Gent, Belgium (e-mail: \href{mailto:pooya.ashtari@esat.kuleuven.be}{pooya.ashtari@esat.kuleuven.be}, \href{mailto:pooya.ashtari@ugent.be}{pooya.ashtari@ugent.be}). Corresponding author.}
    \thanks{Shahryar Noei is with the Data Science for Health Unit, Fondazione Bruno Kessler, Via Sommarive 18, Povo, Trento, Italy (e-mail: \href{mailto:snoei@fbk.eu}{snoei@fbk.eu}).}
    \thanks{Fateme Nateghi Haredasht is with the Center for Biomedical Informatics Research, Stanford University, Stanford, CA, USA (e-mail: \href{mailto:fnateghi@stanford.edu}{fnateghi@stanford.edu}).}
    \thanks{Jonathan H. Chen is with the Center for Biomedical Informatics Research and with the Division of Hospital Medicine, Stanford University, Stanford, CA, USA (e-mail: \href{mailto:jonc101@stanford.edu}{jonc101@stanford.edu}).}
    \thanks{Giuseppe Jurman is with the Data Science for Health Unit, Fondazione Bruno Kessler, Via Sommarive 18, Povo, Trento, Italy and the Department of Biomedical Sciences, Humanitas University, Via Rita Levi Montalcini, 4, 20072 Pieve Emanuele MI (e-mail: \href{mailto:giuseppe.jurman@fbk.eu}{giuseppe.jurman@fbk.eu}).}
    \thanks{Aleksandra Pi\v{z}urica is with the Department of Telecommunications and Information Processing, Ghent University, B-9000 Gent, Belgium (e-mail: \href{mailto:aleksandra.pizurica@ugent.be}{aleksandra.pizurica@ugent.be}).}
    \thanks{Sabine Van Huffel is with the Department of Electrical Engineering (ESAT), STADIUS Center, KU Leuven, Belgium and with Leuven.AI - KU Leuven institute for AI, B-3000, Leuven, Belgium (e-mail: \href{mailto:sabine.vanhuffel@esat.kuleuven.be}{sabine.vanhuffel@esat.kuleuven.be}).}
    \thanks{Pooya Ashtari and Shahryar Noei contributed equally to this work.}
    \thanks{The project is available at \href{https://github.com/pashtari/deconver}{\textcolor{magenta}{https://github.com/pashtari/deconver}}.}
}

\maketitle

\begin{abstract}
    While convolutional neural networks (CNNs) and vision transformers (ViTs) have advanced medical image segmentation, they face inherent limitations such as local receptive fields in CNNs and high computational complexity in ViTs. This paper introduces Deconver, a novel network that integrates traditional deconvolution techniques from image restoration as a core learnable component within a U-shaped architecture. Deconver replaces computationally expensive attention mechanisms with efficient nonnegative deconvolution (NDC) operations, enabling the restoration of high-frequency details while suppressing artifacts. Key innovations include a backpropagation-friendly NDC layer based on a provably monotonic update rule and a parameter-efficient design. Evaluated across four datasets (ISLES'22, BraTS'23, GlaS, FIVES) covering both 2D and 3D segmentation tasks, Deconver achieves state-of-the-art performance in Dice scores and Hausdorff distance while reducing computational costs (FLOPs) by up to 90\% compared to leading baselines. By bridging traditional image restoration with deep learning, this work offers a practical solution for high-precision segmentation in resource-constrained clinical workflows.
\end{abstract}

\section{Introduction} \label{sec:introduction}

\IEEEPARstart{M}edical image segmentation is a fundamental task in modern healthcare, enabling precise delineation of anatomical structures and pathological regions essential for computer-assisted diagnosis, treatment planning, and surgical guidance. Despite advancements, achieving accurate segmentation remains challenging due to inherent complexities of medical images, such as low contrast, heterogeneous textures, and acquisition artifacts such as motion blur or noise.

Convolutional Neural Networks (CNNs), particularly U-Net \cite{ronneberger2015u} and its variants, have dominated medical image segmentation due to their ability to hierarchically extract spatial features. Extensions like 3D U-Net \cite{cciccek2016} and nnU-Net \cite{isensee2021nnu} further improved performance by adapting to volumetric data and automating architecture configurations. However, CNNs are inherently limited by their local receptive fields, hindering their ability to model long-range spatial dependencies, often critical for segmenting anatomically dispersed or structurally complex regions.

Recent efforts to address this limitation include enlarging kernel sizes \cite{roy2023mednext} or adopting Vision Transformers (ViTs). ViT-based models like nnFormer \cite{zhou2021nnformer} and MISSFormer \cite{huang2022missformer} excel at capturing global context via self-attention but suffer from quadratic computational complexity relative to input resolution. This restricts their practicality in high-resolution medical imaging and resource-constrained clinical environments. Hybrid architectures, such as TransUNet \cite{chen2024transunet} and Swin UNETR \cite{hatamizadeh2021swin}, attempt to balance locality and globality by combining convolutional and self-attention layers but often come at the cost of increased architectural complexity.

Deconvolution is a classical technique in image processing widely used for deblurring and image restoration through methods like Wiener filtering \cite{wiener1949extrapolation} and Richardson-Lucy algorithms \cite{lucy1974iterative}, iteratively refine estimates of latent sources under physical constraints like nonnegativity. While effective as a pre-processing step in traditional pipelines, the integration of deconvolution into deep learning frameworks remains underexplored, where it could synergize data-driven feature learning with the image enhancement capability for improved segmentation performance.

In this work, we propose \textbf{Deconver}, a novel segmentation network that integrates deconvolution as a core learnable layer within a U-shaped architecture. Our key insight is to replace computationally expensive attention mechanisms with efficient deconvolution operations, enabling the restoration of high-frequency details while suppressing artifacts. The main contributions of this work are threefold:
\begin{itemize}
    \item \textbf{Architectural Innovation:} Deconver is the first network to incorporate deconvolution principles as a learnable component within a deep architecture.
    \item \textbf{Nonnegative deconvolution layer:} We introduce a backpropagation-friendly, differentiable layer based on a provably monotonic update rule for nonnegative deconvolution, enabling stable end-to-end training using current deep learning frameworks.
    \item \textbf{Performance and efficiency:} Deconver achieves state-of-the-art performance on both 2D and 3D segmentation tasks with substantially fewer computational costs and parameters than leading baselines.
\end{itemize}
Extensive experiments across four datasets (ISLES'22, BraTS'23, GlaS, and FIVES) demonstrate Deconver's superiority in Dice scores and boundary accuracy (Hausdorff distance). By bridging classical image restoration with modern deep learning, Deconver provides a practical solution for high-precision segmentation, particularly in resource-constrained clinical workflows.
\section{Related Work} \label{sec:related_work}

\subsection{Convolutional Neural Networks}
Convolutional Neural Networks (CNNs) have played a central role in medical image segmentation, primarily due to their ability to extract hierarchical spatial features. The introduction of U-Net \cite{ronneberger2015u} established an encoder–decoder architecture with skip connections that has since become the backbone of many segmentation models. Variants have since emerged to improve performance across different settings: 3D U-Net \cite{cciccek2016} extends U-Net to volumetric data using 3D operations; UNet++ \cite{zhou2018unet++} incorporates nested dense skip connections; and nnU-Net \cite{isensee2021nnu} presents a self-adapting framework capable of configuring itself to a wide range of tasks. Other notable CNN-based advances include SegResNet \cite{myronenko20193d}, a residual U-Net variant that won the Brain Tumor Segmentation Challenge (BraTS) 2018.

Despite their effectiveness, a key limitation of CNNs is their inherently local receptive field, which restricts their capacity to capture long-range spatial dependencies. This poses challenges for segmenting anatomically complex or spatially dispersed structures. Two major approaches have emerged to address this issue. One involves using large  kernels, as seen in MedNeXt \cite{de2024robust}, which expands the receptive field by iteratively increase kernel sizes by upsampling small kernel networks. The other approach incorporates attention mechanisms to explicitly model global context, paving the way for Transformer-based and hybrid segmentation architectures.

\subsection{Transformers}
Originally introduced for natural language processing, Transformer architectures have been successfully adapted to computer vision tasks through Vision Transformers (ViTs) \cite{dosovitskiy2020image}, which model images as sequences of patch tokens. Fully Transformer-based segmentation models incorporate self-attention mechanisms in both the encoder and decoder, offering a shift from traditional convolution-based designs.

Among these, nnFormer  \cite{zhou2021nnformer} proposes an interleaved architecture that combines local and global self-attention layers with convolutional downsampling. MISSFormer  \cite{huang2022missformer} builds a hierarchical Transformer-based encoder-decoder tailored for medical image segmentation, enhancing both local precision and long-range contextual reasoning. Self-attention mechanisms are known to be computationally intensive, especially on long sequences, which can hinder their scalability in high-resolution medical imaging tasks. One solution has been proposed in  \cite{ashtari2023factorizer} by replacing attention with non-negative matrix factorization (NMF) which was shown to reduce the computational cost significantly while maintaining high performance. Other approaches include using hybrid models.

\subsection{Hybrid Models}
To balance the strengths of convolutional and attention-based methods, hybrid architectures have emerged as a practical solution. These models typically use a Transformer in the encoder and adopt a CNN-based decoder. One of the earliest hybrid models in medical imaging, TransUNet  \cite{chen2024transunet}, incorporates a ViT encoder into the bridge of a U-Net. Extending this work, UNETR  \cite{hatamizadeh2022unetr} employs a full Transformer-based encoder directly connected to a convolutional decoder via skip connections. Swin UNETR  \cite{hatamizadeh2021swin} further improves upon this by replacing the ViT encoder with Swin Transformer blocks  \cite{liu2021swin}, introducing a hierarchical structure that models both local and global dependencies efficiently across scales. While hybrid models are effective they still come with increased complexity and computational demands.

\subsection{Deconvolution}
Deconvolution is a fundamental technique in image processing aimed at reversing the effects of blurring and restoring an image closer to its original form \cite{satish2020comprehensive}. Early methods such as Wiener deconvolution \cite{wiener1949extrapolation} applied frequency-domain filtering to restore degraded signals, while iterative approaches like the Richardson-Lucy algorithm \cite{richardson1972bayesian, lucy1974iterative} refined the image estimate through successive likelihood-based updates.

Deconvolution techniques have been applied in medical image analysis to enhance image quality by mitigating blurring effects inherent in various imaging modalities. In fluorescence microscopy, deconvolution algorithms are employed to restore high-resolution details from blurred images, thereby improving the visualization of cellular structures \cite{katoh2024recent}. Similarly, in magnetic resonance imaging (MRI) \cite{debnath2013deblurring}, computed tomography (CT) \cite{sharma2016mse,liu2008deconvolution} and positron emission tomography (PET) \cite{sample2024neural}, deconvolution methods are utilized to mitigate blurring and improve image resolution. Despite the widespread use of deconvolution in traditional medical image processing, its full integration into deep learning architectures for improved segmentation has not been explored.

\section{Methods} \label{sec:methods}

\begin{figure*}[!t]
    \centering
    \includegraphics[width=.9\textwidth]{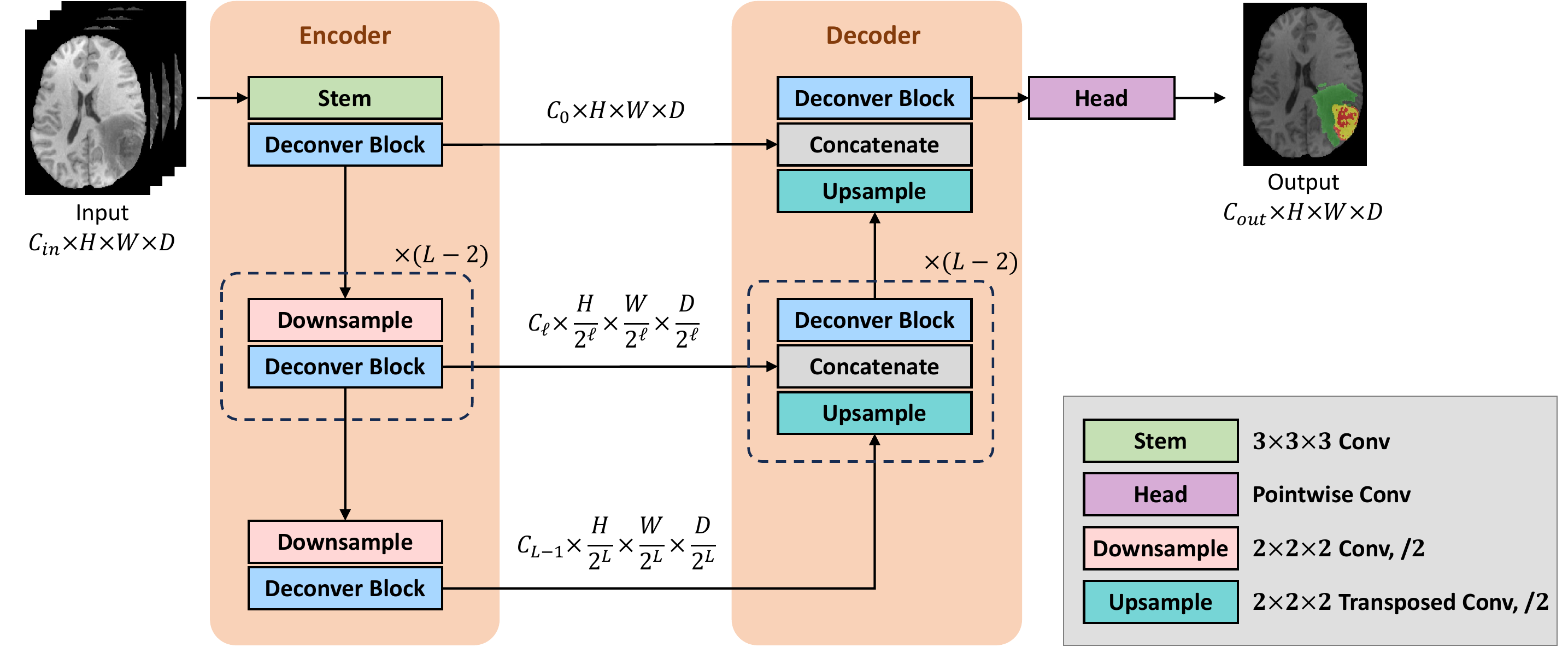}
    \vspace{1em}
    \caption{Overview of Deconver architecture.}
    \label{fig:deconver}
\end{figure*}

\subsection{Notation} \label{sec:notation}

We denote vectors by boldface lowercase letters (e.g., \( \vct{x}\)), matrices by boldface uppercase letters (e.g., \( \mtx{X} \)), and tensors by boldface calligraphic letters  (e.g., \( \tns{X} \)). For clarity, a 2D input image is represented as a 3D tensor \( \tns{X} \in \R^{C \times H \times W} \), where \(C\) denotes the number of channels, and \(H\) and \(W\) represent the spatial height and width. A 3D volumetric input is represented as a 4D tensor \( \tns{X} \in \R^{C \times H \times W \times D}\), with \( D \) denoting depth. Individual elements in a tensor are accessed via indices matching its dimensions, such as \( \tns{X}[c, h, w]\) for a 3D tensor or \( \tns{X}[c, h, w, d]\) for a 4D tensor.  

The inner product between two tensors \( \tns{X} \) and \( \tns{Y} \) of identical dimensions is denoted by 
\begin{equation*} \label{eq:inner_product}
    \langle \tns{X}, \tns{Y} \rangle = \sum_{i_1, \ldots, i_N} \tns{X}[i_1, \ldots, i_N] \ \tns{Y}[i_1, \ldots, i_N],
\end{equation*}
where \(N\) is the number of tensor dimensions. The Frobenius norm is defined as \( \| \tns{X} \|_\text{F} = \sqrt{\langle \tns{X}, \tns{X}  \rangle} \).

\subsection{Revisiting Deconvolution} \label{sec:deconvolution}

Deconvolution is a fundamental technique in image processing that seeks to reverse the effects of convolution to restore images degraded by blurring, noise, or other distortions. In medical image segmentation, deconvolution is particularly valuable for enhancing fine anatomical details and mitigating acquisition artifacts. This enables more precise delineation of structures such as tumors and blood vessels by recovering high-frequency components often lost during image acquisition.

\begin{figure*}[!t]
    \centering
    \begin{subfigure}{0.3\textwidth}
        \centering
        \includegraphics[width=.55\textwidth]{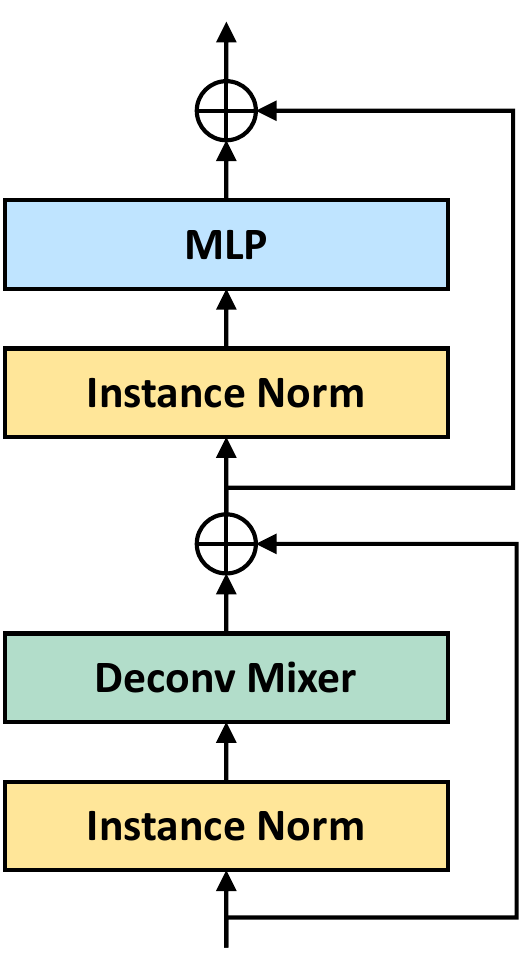}
        \caption{Deconver block}
        \label{fig:deconver_block}
    \end{subfigure}%
    \begin{subfigure}{0.3\textwidth}
        \centering
        \includegraphics[width=.55\textwidth]{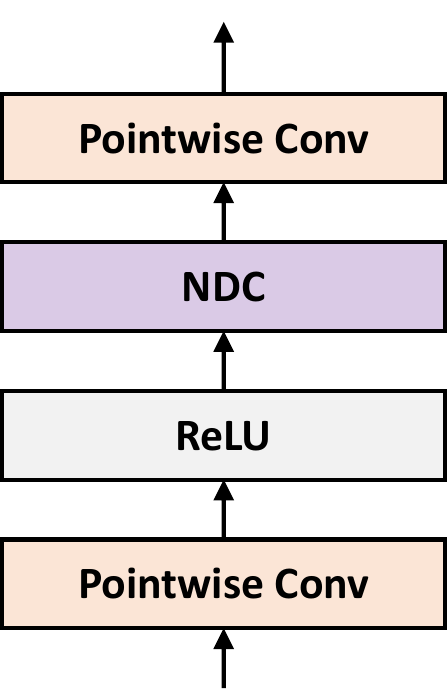}
        \caption{Deconv Mixer}
        \label{fig:deconv_mixer}
    \end{subfigure}%
    \begin{subfigure}{0.3\textwidth}
        \centering
        \includegraphics[width=.7\textwidth]{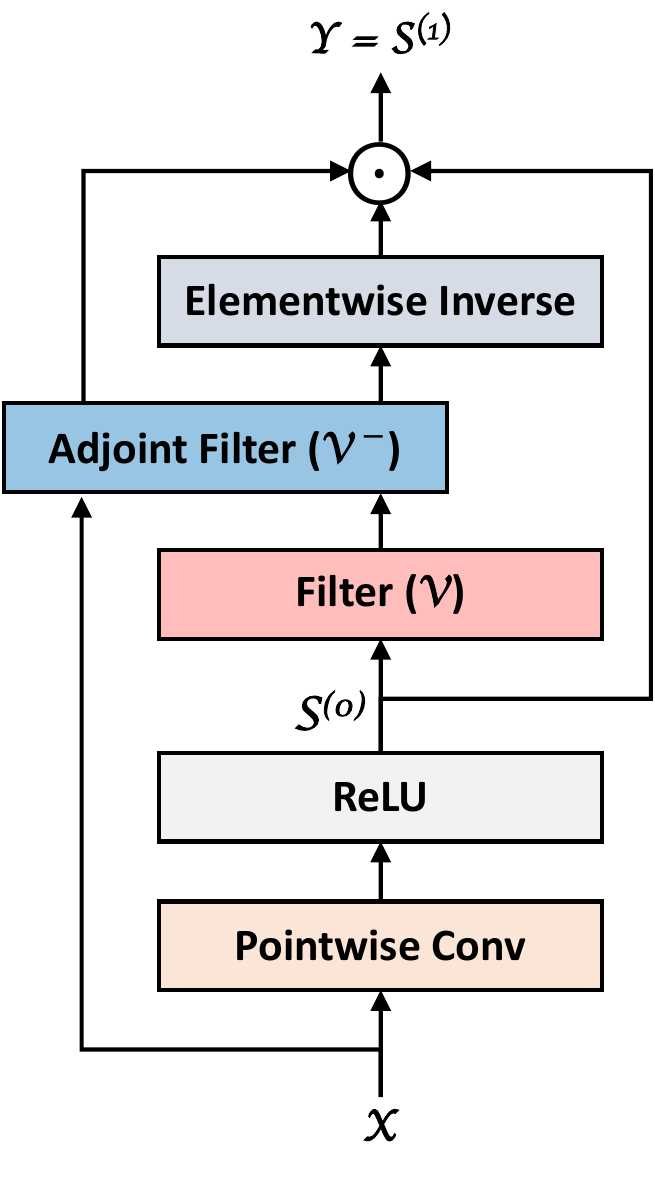}
        \caption{NDC layer}
        \label{fig:ndc}
    \end{subfigure}

    \vspace{1em}

    \caption{Overview of Deconver block and its components.}
    \label{fig:deconver_components}
\end{figure*}

\paragraph{Problem Formulation} Let a 2D input image be represented as \( \tns{X} \in \R^{C \times H \times W} \), where \( C \) denotes the number of channels, and \( H \) and \( W \) represent the height and width of the image, respectively. The objective is to recover the latent \emph{source image} \( \tns{S} \in \R^{E \times H \times W} \) from the observed \( \tns{X} \), given a known \emph{filter tensor} \( \tns{V} \in \R^{C \times E \times (2M+1) \times (2N+1)} \), where \( E \) represents the number of channels of the source image, and \( (2M+1, 2N+1) \) is the spatial size of the filter. Deconvolution aims to approximate  \( \tns{X} \) as the cross-correlation of \( \tns{S} \) and \( \tns{V} \), i.e., \( \tns{X} \approx \hat{\tns{X}} = \tns{S} \ast \tns{V} \), defined as
\begin{equation} \label{eq:cross-correlation}
    \hat{\tns{X}}[c,h,w] \triangleq \sum_{e=0}^{E-1} \sum_{m=0}^{2M} \sum_{n=0}^{2N} \tns{S}_{p}[e,h+m,w+n] \tns{V}[c,e,m,n],
\end{equation}
for \( c \in \{0, \ldots, C-1\} \), \( h \in \{0, \ldots, H-1\} \), and \( w \in \{0, \ldots, W-1\} \). Here, \( \tns{S}_p = \text{pad}(\tns{S}, (M,N)) \in \R^{E \times (H+2M) \times (W+2N)} \) denotes the zero-padded source image, ensuring to preserve spatial dimensions post-filtering (note that this definition aligns with CNN conventions, applying the filter without flipping). This formulation extends standard convolution by allowing multi-channel inputs and outputs, making it suitable for modern deep learning architectures.

\paragraph{Nonnegative Deconvolution (NDC)} In this work, we focus on the \emph{nonnegative deconvolution (NDC)}, where \( \tns{X} \geq 0 \), \( \tns{V} \geq 0 \), and \( \tns{S} \geq 0 \). The nonnegativity constraint aligns with the physical nature of medical imaging systems, where intensities are inherently positive, helping suppress negative artifacts that could mislead segmentation models. The goal is to estimate the source image \( \tns{S} \geq 0 \) by minimizing the reconstruction error:
\begin{align} \label{eq:deconv_problem}
    \minimize_{\tns{S}} & \quad \mathcal{E}(\tns{S}) = \| \tns{X} - \tns{S} \ast \tns{V} \|_\text{F}^2 \\ \nonumber \
    \text{subject to}    & \quad \tns{S} \geq 0,
\end{align}
where \( \|\cdot\|_\text{F} \) denotes the Frobenius norm. This formulation implicitly assumes additive Gaussian noise, which approximates complex noise in many imaging modalities while ensuring computational tractability.

\paragraph{Multiplicative Update Rule} To address problem \eqref{eq:deconv_problem}, we can derive an iterative update rule inspired by Richardson-Lucy algorithm \cite{lucy1974iterative} and nonnegative matrix factorization \cite{lee2000algorithms}. Starting with an initial guess \( \tns{S}^{(0)} \geq 0 \), the source image at iteration \( t+1 \) is updated as:
\begin{align} \label{eq:update_source}
    \tns{S}^{(t+1)} = \tns{S}^{(t)} \odot \frac{\tns{X} \ast \tns{V}^{-}}{(\tns{S}^{(t)} \ast \tns{V}) \ast \tns{V}^{-}},
\end{align}
where \( \odot \) denotes element-wise multiplication, \( \tns{V}^{-} \) is the \emph{adjoint filter}, which is transposed and spatially flipped (i.e., \(\tns{V}^{-}[d,c,m,n] = \tns{V}[c,d,2M-m,2N-n]\)), and the division is element-wise. The numerator correlates residuals with the filter, amplifying regions where \( \tns{S} \) underestimates \( \tns{X} \), while the denominator normalizes the update to prevent overshooting. This multiplicative form inherently preserves nonnegativity when \(\tns{S}^{(0)} \geq 0\).

\paragraph{Monotonicity} The key advantage of the multiplicative update \eqref{eq:update_source} is its guarantee of the reduction of the reconstruction error, which we will prove in theorem \ref{the:monotonicity}.

\begin{theorem}[Monotonicity] \label{the:monotonicity}
    Let \( \tns{S}^{(t)} \) be the source image at iteration \( t \). With a nonnegative initial source \( \tns{S}^{(0)} \geq 0 \) and under the update \eqref{eq:update_source}, the reconstruction error \( e^{(t)} \triangleq \|\tns{X} - \tns{S}^{(t)} \ast \tns{V}\|_\text{F}^2 \) is non-increasing, i.e., \( e^{(t+1)} \leq e^{(t)} \) for all \( t \geq 0 \).
\end{theorem}
\begin{proof}
    See Appendix for a detailed proof.
\end{proof}

The update rule \eqref{eq:update_source} also generalizes naturally to 3D volumes, making it suitable for modalities like MRI and CT. In Section \ref{sec:ndc_layer}, we integrate this deconvolution technique as a learnable layer within a deep neural network, enhancing multi-scale feature maps to improve segmentation performance.

\subsection{Overall Architecture} \label{sec:overall_architecture}

The Deconver architecture adopts a U-shaped structure (see Fig. \ref{fig:deconver}), comprising an encoder and decoder with skip connections in between at equal resolutions. Given a 3D input image \( \tns{X} \in \R^{C_{in} \times H \times W \times D} \) with \( C_{in} \) input channels and spatial dimensions \((H, W, D)\), the network outputs a \emph{logit} map of shape \( (C_{out}, H, W, D) \), where \( C_{out} \) denotes the number of target classes.

The encoder consists of \( L \) stages, each containing a \emph{Deconver block} described in Section \ref{sec:deconver_block}. At the initial stage, a \emph{stem} layer increases the input channels to $ C_0 $ (typically 32 or 64, depending on the dataset) using a single convolutional layer with a kernel size of \( (3,3,3) \). Subsequent stages downsample feature maps using strided convolutions (stride=2), halving the spatial dimensions while doubling the channel count until a maximum of 512 channels is reached. Formally, the number of output channels at stage \( \ell \) is set to \( C_\ell = \min(C_0 \times 2^\ell, 512) \). This design balances computational efficiency with the capacity to learn abstract, high-level representations. As the encoder deepens, the growing receptive field enables the capture of global contextual relationships essential for distinguishing semantically similar but spatially distant structures.

The decoder mirrors the encoder’s hierarchical structure but reverses the spatial reduction via transposed convolutions (stride=2) to upsample feature maps. At each stage, the decoder incorporates skip connections that concatenate upsampled features with their encoder counterparts at corresponding resolutions. These skip connections mitigate information loss during downsampling and enhance feature reuse, ensuring more precise localization of fine-grained structures. At the deepest decoder layer, a pointwise convolution head (\( 1 \times 1 \times 1 \) kernel) generates the final \emph{logit} map, which can be activated via sigmoid or softmax to produce class probabilities for segmentation.

\subsection{Deconver Block} \label{sec:deconver_block}

A Deconver block forms the main building unit of the proposed model. In contrast to Vision Transformer (ViT) blocks \cite{dosovitskiy2020image} that rely on attention mechanism, our Deconver block replaces the multi-head self-attention module with a learnable \emph{Deconv Mixer} module (presented in Section \ref{sec:deconv_mixer}), and substitutes layer normalization with instance normalization \cite{ulyanov2016instance} to better accommodate the small batch sizes often used with 3D or high-resolution medical images.

As illustrated in Fig. \ref{fig:deconver_block}, the block consists of two sequential sub-modules: \emph{Deconv Mixer} and Multi-Layer Perceptron (MLP). Each sub-module is preceded by instance normalization and followed by a residual connection. Formally, given an intermediate feature map \( \tns{X} \in \R^{C_\ell \times H \times W \times D} \) at stage \( \ell \), the block’s operations are defined as
\begin{align}
    \label{eq:deconver_block}
     & \tns{Z} = \text{DeconvMixer}(\text{InstanceNorm}(\tns{X})) + \tns{X}, \nonumber \\
     & \tns{Y} = \text{MLP}(\text{InstanceNorm}(\tns{Z})) + \tns{Z},
\end{align}
where the MLP is composed of two pointwise convolution layers separated by a Gaussian Error Linear Unit (GELU) activation:
\begin{equation}
    \label{eq:mlp}
    \text{MLP}(\tns{X}) = \text{PointwiseConv}(\text{GELU}(\text{PointwiseConv}(\tns{X}))).
\end{equation}
The MLP expands the channel dimension by a factor of \( \alpha \) before projecting back to the original dimension, enabling nonlinear interaction across channels while preserving spatial structure.

\subsection{Deconv Mixer} \label{sec:deconv_mixer}

As illustrated in Fig. \ref{fig:deconv_mixer}, the Deconv Mixer module processes input features through three sequential stages: an initial pointwise convolution, an \emph{NDC layer}, and a final pointwise convolution.

First, Deconv Mixer applies a pointwise convolution to linearly project each position. The output is then passed through a ReLU activation function, enforcing nonnegativity to ensure compatibility with the subsequent \emph{NDC layer} (described in Section \ref{sec:ndc_layer}). This nonnegative feature map is then processed by the NDC layer, which captures spatial dependencies and restores high-frequency details that may have been lost in previous layers. Finally, a second pointwise convolution is applied to produce the output. Formally, given an input feature map \( \tns{X} \), the Deconv Mixer can be expressed as:
\begin{align}
    \label{eq:deconv_mixer}
     & \tns{X}^1 = \text{PointwiseConv}(\tns{X}), \nonumber           \\
     & \tns{X}^2 = \text{NDC}(\text{ReLU}(\tns{X}^1)), \nonumber      \\
     & \text{DeconvMixer}(\tns{X}) = \text{PointwiseConv}(\tns{X}^2),
\end{align}
where the intermediate tensors \( \tns{X}^1 \) and \( \tns{X}^2 \), and the final output share the same shape as the input \( \tns{X} \).

\subsection{Nonnegative Deconvolution Layer} \label{sec:ndc_layer}

The nonnegative deconvolution (NDC) layer forms the core innovation of Deconver, incorporating nonnegative deconvolution (presented in Section \ref{sec:deconvolution}) as a learnable layer to enhance feature representations.

The NDC layer operates in a grouped manner, partitioning an input feature map \( \tns{X} \in \R_{\geq 0}^{C \times H \times W \times D} \) into \( G \) groups \( \{\tns{X}_g\}_{g=1}^G \) along the channel dimension. Each group \( \tns{X}_g \in \R_{\geq 0}^{C/G \times H \times W \times D} \) is processed independently, maintaining its own learnable filter \( \tns{V}_g \geq 0 \) and source image \( \tns{S}_g \geq 0 \). 

The layer introduces a source channel ratio \( R \), defined as \( R = E/C \), where \( E \) denotes the number of source channels. This ratio controls the channel expansion of the source image relative to the input.

The initial source image \( \tns{S}_g^{(0)} \in \R_{\geq 0}^{E/G \times H \times W \times D} \) is derived from the input feature map \( \tns{X} \) through a pointwise convolution followed by ReLU activation (See Fig. \ref{fig:ndc}). This ensures nonnegativity while providing an adaptive and learnable initialization of the source. The filter \( \tns{V}_g \in \R_{\geq 0}^{C/G \times R C/G \times (2M+1) \times (2N+1)} \) is initialized using the Kaiming uniform distribution \cite{he2015delving} and clamped to nonnegative values via ReLU before being plugged into the update rule.

The NDC layer applies a single iteration of the multiplicative update rule \eqref{eq:update_source} to refine the source image. For computational efficiency, we empirically found one iteration sufficient to achieve a good trade-off between accuracy and computational cost. The update for group \( g \) is:
\begin{equation}\label{eq:update_source_ndc}
    \tns{S}_g^{(1)} = \tns{S}_g^{(0)} \odot \frac{\tns{X}_g * \tns{V}_g^{-} + \epsilon}{\left(\tns{S}_g^{(0)} * \tns{V}_g\right) * \tns{V}_g^{-} + \epsilon},
\end{equation}
where \( \epsilon =  10^{-8}\) is a small positive constant to avoid division by zero.  The numerator amplifies regions where the source underestimates the input, while the denominator normalizes the update to prevent overshooting. The final output is the channel-wise concatenation of all group outputs \( \{\tns{S}_g^{(1)}\}_{g=1}^G \), preserving spatial resolution and expanding channel dimensions by \(R\). The NDC layer is fully differentiable and backpropagation-friendly, enabling end-to-end training using modern deep learning frameworks. Additionally, the filter \( \tns{V}_g \) and its adjoint \( \tns{V}_g^{-} \) share the same learnable parameters, reducing parameter overhead and potentially improving generalization performance.
\section{Experiments} \label{sec:experiments}

This section details the experimental setup, comparative analyses, and ablation studies conducted to assess Deconver's effectiveness in 2D and 3D, as well as binary and multi-class medical image segmentation tasks. We evaluate its performance across four datasets, benchmark it against state-of-the-art baselines, and analyze the impact of key architectural design choices.

\subsection{Experimental Setup} \label{sec:setup}

\subsubsection{Datasets} \label{sec:datasets}

We evaluated Deconver on four publicly available datasets, covering both 3D (ISLES'22 and BraTS'23) and 2D (GlaS and FIVES) medical imaging modalities:
\begin{itemize}
    \item \textbf{ISLES'22} \cite{ hernandez2022isles}: This dataset includes 250 multi-center MRI scans, targeting ischemic stroke lesions via diffusion-weighted imaging (DWI) and apparent diffusion coefficient (ADC) maps. We excluded FLAIR images to simplify the pipeline and avoid registration challenges.
    \item \textbf{BraTS'23} \cite{menze2015multimodal, baid2021rsna, bakas2017advancing}: This dataset consists of 1,251 multi-parametric MRI (mpMRI) scans, including native T1-weighted, post-Gadolinium T1-weighted, T2-weighted, and FLAIR sequences. Ground truth labels delineate three tumor subregions: enhancing tumor (ET), tumor core (TC), and whole tumor (WT).
    \item \textbf{GlaS} \cite{sirinukunwattana2017gland, sirinukunwattana2015stochastic}: Containing 165 high-resolution H\&E-stained colorectal histopathology images, this dataset features expert-annotated gland segmentations.
    \item \textbf{FIVES} \cite{jin2022fives}: This dataset includes 800 high-resolution fundus photographs with manual segmentation of retinal blood vessels.
\end{itemize}

\subsubsection{Baseline Models} \label{sec:baselines}

We compare Deconver against several state-of-the-art baseline models, including CNNs such as nnU-Net \cite{isensee2021nnu} and SegResNet \cite{myronenko20193d}; hybrid convolution-transformer architectures like UNETR \cite{hatamizadeh2022unetr} and Swin UNETR \cite{hatamizadeh2021swin}; and Factorizer \cite{ashtari2023factorizer}, which leverages non-negative matrix factorization (NMF). These baselines represent diverse approaches to medical image segmentation.

\subsubsection{Implementation Details} \label{sec:impelementation_details}

All models were implemented using PyTorch and the MONAI framework, trained on a single NVIDIA H100 GPU. We used the AdamW optimizer with an initial learning rate of 0.0001 and weight decay of 0.00001. A cosine annealing learning rate scheduler was used, incorporating a 1\% warm-up phase during which the learning rate was scaled by 10. Training was conducted for 500 epochs for ISLES'22, GlaS, and FIVES, and 300 epochs for BraTS'23. The batch size was set to 2 for BraTS'23 and 8 for all other datasets. Random patches were extracted during training, with sizes of \( (64, 64, 64) \) for ISLES'22, \( (128, 128, 128) \) for BraTS'23, \( (256, 256, 256) \) for GlaS, and \( (512, 512, 512) \) for FIVES. Data augmentation included random affine transformations, flipping, Gaussian noise, Gaussian smoothing, and intensity scaling/shifting. We use the sum of soft Dice \cite{milletari2016v} and cross-entropy losses as the training objective. For inference, a patch-based sliding window approach with a 50\% overlap and the same patch size as training was adopted. The final binary segmentation maps were obtained by thresholding predicted probabilities (sigmoid outputs).

\subsubsection{Model Configurations} \label{sec:model_configs}

Deconver was configured with dataset-specific hyperparameters to balance model capacity and computational efficiency. The encoder depth (\( L \)) was set to 4 for ISLES'22, 5 for BraTS'23, and 6 for GlaS and FIVES. The base number of channels (\( C_0 \)) set to 64 for ISLES'22, and 32 for BraTS'23, GlaS, and FIVES. At each encoder stage (\( \ell \)), the channel dimension (\( C_\ell \)) was determined by the formula \( C_\ell = \min(C_0 \times 2^\ell, 512) \), doubling the channels after each downsampling step until reaching a maximum of 512 channels.

In the NDC layers, the number of groups (\( G \)) was set equal to the number of input channels by default. Inspired by depthwise separable convolutions and validated through ablation studies (Section \ref{sec:ablation_studies}), this design choice reduces computational complexity while preserving the ability to model diverse spatial patterns. Additionally, the source channel ratio (\( R \)) was fixed at 4, as our experiments found this to optimally balance accuracy and efficiency. The MLP expansion factor (\( \alpha \)) was fixed at 4 across all experiments.

\subsubsection{Evaluation Metrics} \label{sec:metrics}

We performed stratified 5-fold cross-validation to assess generalization to unseen data. Segmentation performance was quantified using two metrics: Dice Similarity Coefficient (DSC) and Hausdorff Distance 95th percentile (HD95). The DSC is defined as:
\begin{equation}
    \label{eq:dice}
    \text{DSC}(\vct{g}, \vct{y}) = \frac{2 \sum_{n=1}^{N} \vct{g}[n]\,\vct{y}[n]}{\sum_{n=1}^{N} \vct{g}[n] + \sum_{n=1}^{N} \vct{y}[n]},
\end{equation}
where \( \vct{g}[n], \vct{y}[n] \in \{0,1\} \) denote the ground truth and predicted labels for voxel \( n \), respectively, and \( N \) represents the total number of voxels. The DSC is defined as 1 when both the ground truth and the prediction contain only zeros.
. Hausdorff Distance is computed as:
\begin{equation}
    \label{eq:hd95}
    \text{HD}(\mathbb{G}, \mathbb{Y}) = \max \left\{ \max_{\vct{g} \in \mathbb{G}} \min_{\vct{y} \in Y} d(\vct{g},\vct{y}), \max_{\vct{y} \in \mathbb{Y}} \min_{\vct{g} \in \mathbb{G}} d(\vct{g},\vct{y}) \right\}.
\end{equation}
where \( d(\vct{g},\vct{y}) \) represents the Euclidean distance between points \(\vct{g}\) and \(\vct{y}\); and \( \mathbb{G} \) and \( \mathbb{Y} \) are sets of all pixel (or voxel) positions on the surface of the ground truth and prediction, respectively. The HD95 metric computes the 95th percentile of distances rather than the maximum, providing a more robust measure against outliers. All the results are reported as the average over the 5-fold cross-validation.

\subsection{Results: 3D Segmentation} \label{sec:results_3d}

\subsubsection{Ischemic Stroke Lesions (ISLES'22)} \label{sec:results_isles}

Table \ref{tab:quantitative_results_isles} presents the quantitative results for ISLES'22. Both variants of Deconver outperform all the baselines in terms of DSC, with the version using a kernel size of \(3 \times 3 \times 3 \) achieving the highest DSC (78.16\%) followed closely by the variant using a larger kernel size of \(5 \times 5 \times 5 \) (77.37\%). In terms of boundary delineation accuracy, measured by the HD95 metric, Deconver demonstrated superior results, with HD95 values of 4.99 and 4.89 for the \( 3 \times 3 \times 3 \) and \( 5 \times 5 \times 5 \) kernels, respectively.

\begin{table}[!t]
    \caption{Segmentation performance comparison on ISLES'22. Best results are \textbf{bold}, second-best are \underline{underlined}.}
    \label{tab:quantitative_results_isles}
    \centering
    \renewcommand{\arraystretch}{1.2}
    \resizebox{\linewidth}{!}{
        \begin{tabular}{l | c | c | c | c}
            \toprule
            \textbf{Model}                     & \textbf{Params} & \textbf{FLOPs / voxel} & \textbf{DSC (\%)} & \textbf{HD95}    \\
            \midrule
            nnU-Net                            & 22.4M           & 3423.5K                & 76.76             & 5.54             \\
            SegResNet                          & 75.9M           & 2228.6K                & 76.85             & 5.18             \\
            UNETR                              & 133.2M          & 525.9K                 & 73.74             & 6.54             \\
            Swin UNETR                         & 62.2M           & 4356.8K                & 76.58             & 5.55             \\
            Factorizer                         & 7.5M            & 2266.7K                & 76.73             & 5.93             \\
            \midrule
            Deconver (3\(\times\)3\(\times\)3) & 10.5M           & 607.0K                 & \textbf{78.16}    & \underline{4.99} \\
            Deconver (5\(\times\)5\(\times\)5) & 11.0M           & 607.0K                 & \underline{77.37} & \textbf{4.89}    \\
            \bottomrule
        \end{tabular}
    }
\end{table}

Notably, both variants of Deconver significantly reduce computational complexity compared to the best-performing baselines. Deconver requires over 70\% fewer FLOPs per voxel compared to the SegResNet, the best performing baseline.  Additionally, Deconver uses around 85\% fewer parameters, resulting in a highly compact architecture without sacrificing segmentation performance. Refer to Fig. \ref{fig:dice_comparison} for a comparison of model performance versus computational efficiency on ISLES'22.

\begin{figure*}[!t]
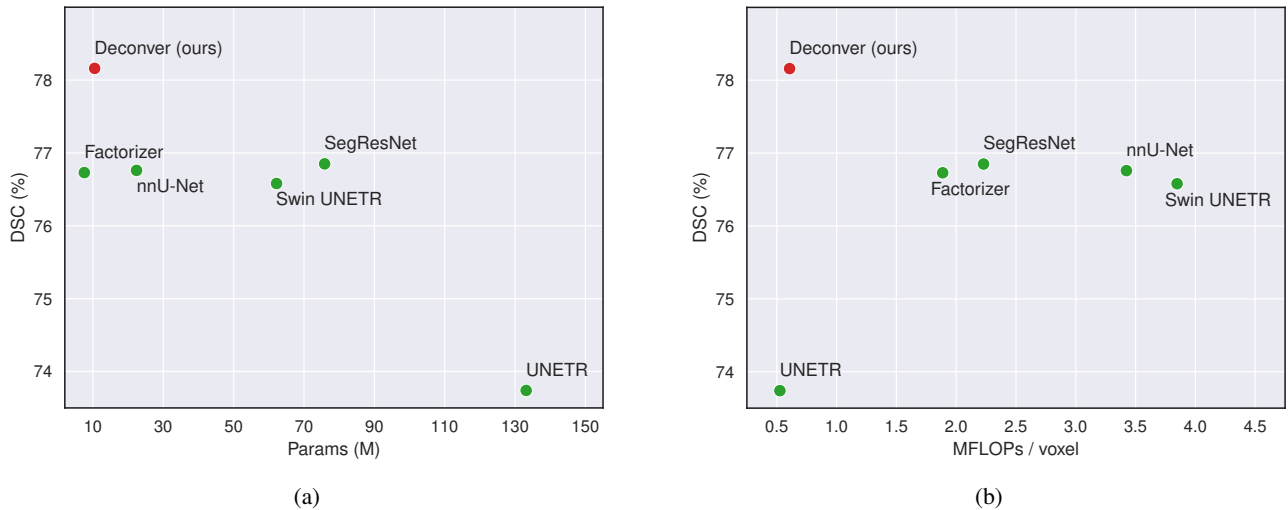

    \centering
    \begin{subfigure}{0.5\textwidth}
        \centering
        \resizebox{.9\textwidth}{!}{\input{figures/dice_params.pgf}}
        \caption{}
        \label{fig:dice_params}
    \end{subfigure}%
    \begin{subfigure}{0.5\textwidth}
        \centering
        \resizebox{.9\textwidth}{!}{\input{figures/dice_flops.pgf}}
        \caption{}
        \label{fig:dice_flops}
    \end{subfigure}
    \caption{Comparison of Dice against the number of parameters (left) and FLOPs/voxel (right) for different models on ISLES'22. Deconver (the variant using \(3 \times 3 \times 3\) kernel) manages to maintain the highest DSC with fewer parameters and FLOPs/voxel.}
    \label{fig:dice_comparison}
\end{figure*}

We also qualitatively examined the segmentation results, as shown in Fig. \ref{fig:qualitative_isles}. The figure presents a visual comparison of a representative slice from the ISLES'22 dataset, showcasing the predictions of different models. For Deconver, we illustrate the results form the \( 3 \times 3 \times 3 \) kernel variant. As evident from the figure, Deconver provides superior segmentation quality compared to other baselines, achieving a more accurate delineation of the lesion while maintaining minimal false positives and false negatives.

The qualitative results show that nnU-Net, Swin UNETR, and UNETR undersegment the lesion in both examples, leading to clinically relevant false negatives. SegResNet captures the lesion more completely but has introduced a false-positive area in the first example marked by the orange circle.

\begin{figure*}[!t]
    \captionsetup[subfigure]{aboveskip=0.25em, belowskip=0.5em, font={bf,small}}
    \centering

    \begin{tikzpicture}
        \fill[blue] (0,0) rectangle (0.4,0.4);
        \node[right] at (0.4,0.2) {True Positives};

        \fill[green] (4,0) rectangle (4.4,0.4);
        \node[right] at (4.4,0.2) {False Positives};

        \fill[red] (8,0) rectangle (8.4,0.4);
        \node[right] at (8.4,0.2) {False Negatives};
    \end{tikzpicture}

    \vspace{5pt}

    \begin{subfigure}{0.2\textwidth}
        \includegraphix[.98\textwidth]{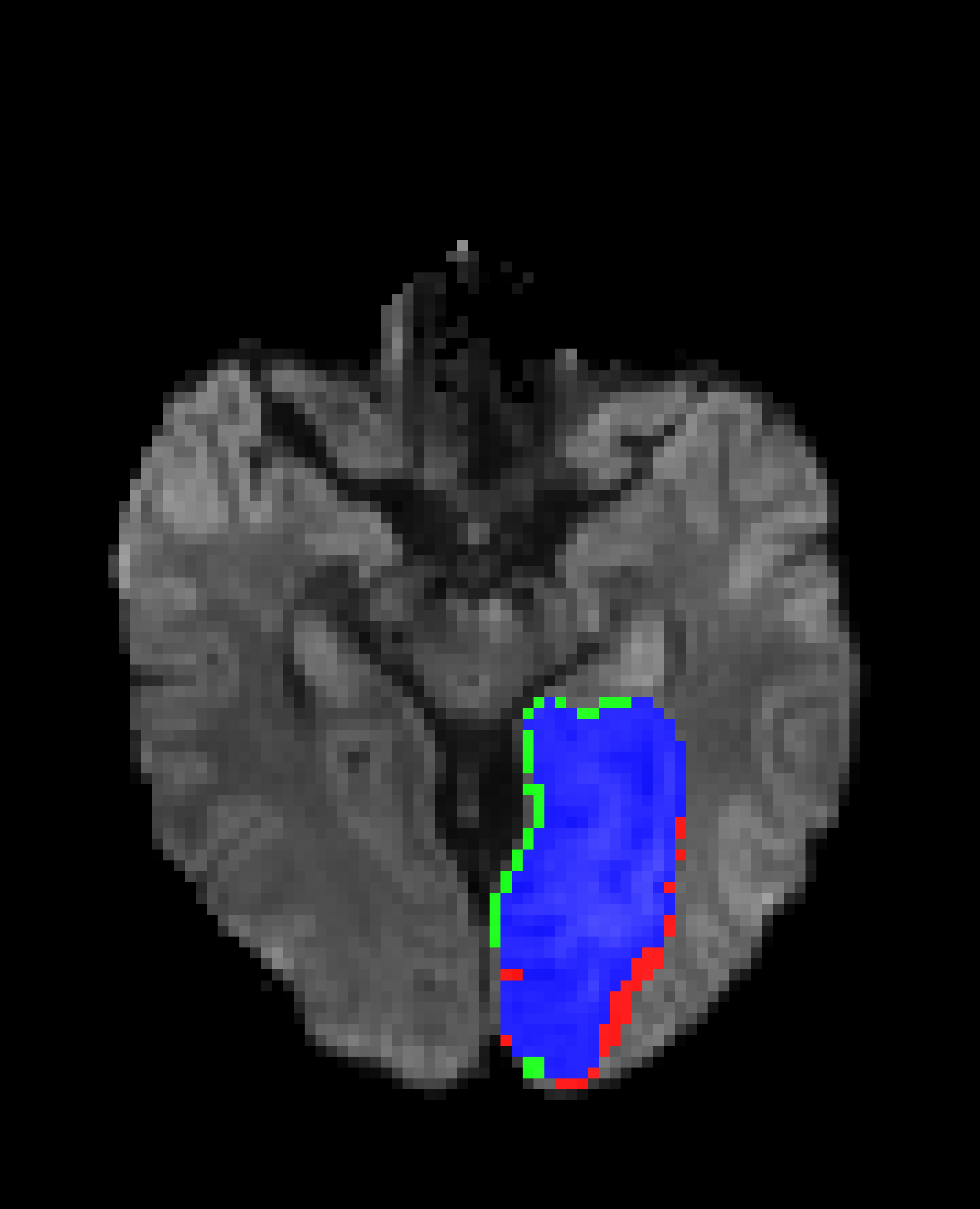}{78.3}
    \end{subfigure}%
    \begin{subfigure}{0.2\textwidth}
        \includegraphix[.98\textwidth]{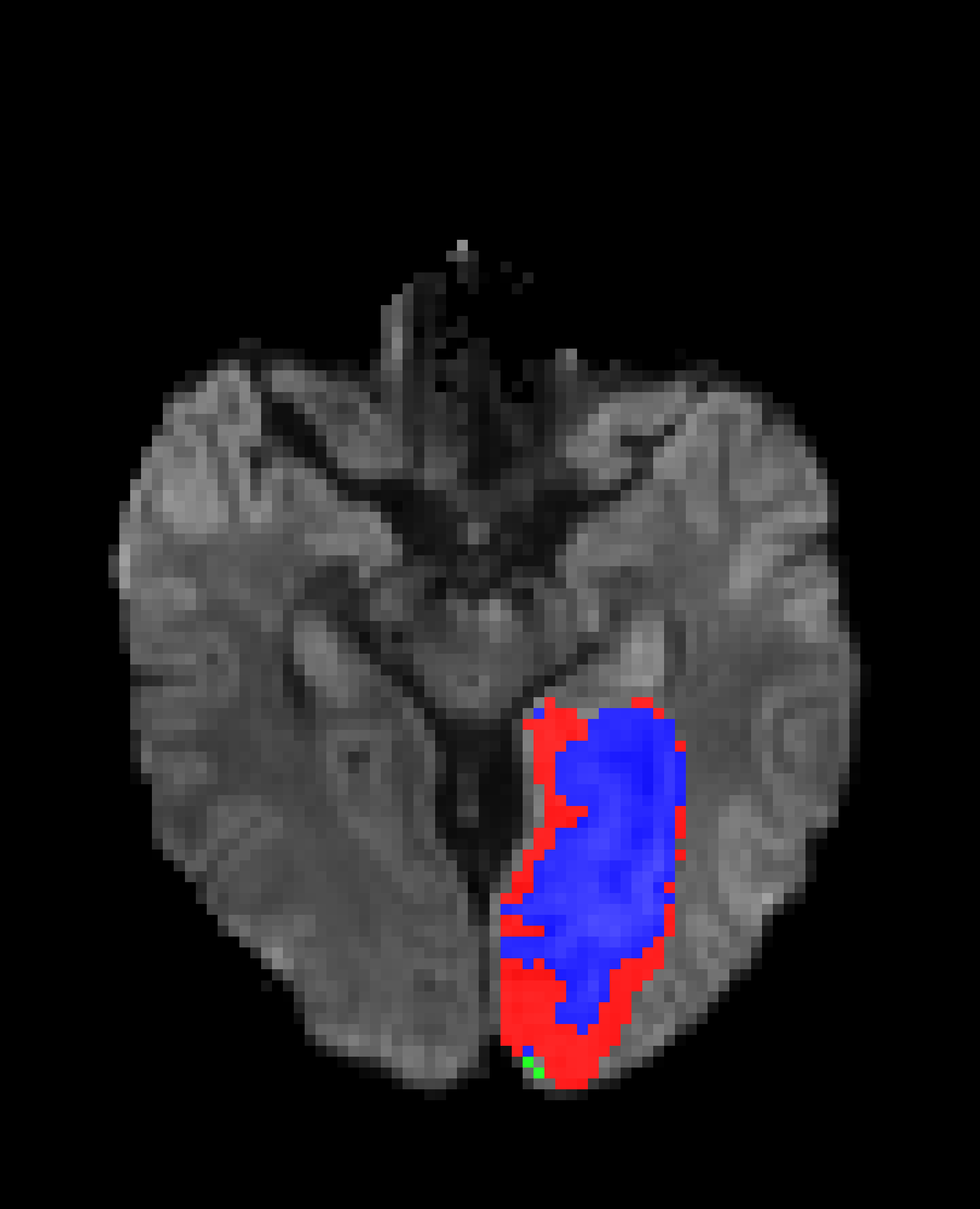}{57.4}
    \end{subfigure}%
    \begin{subfigure}{0.2\textwidth}
        \includegraphix[.98\textwidth]{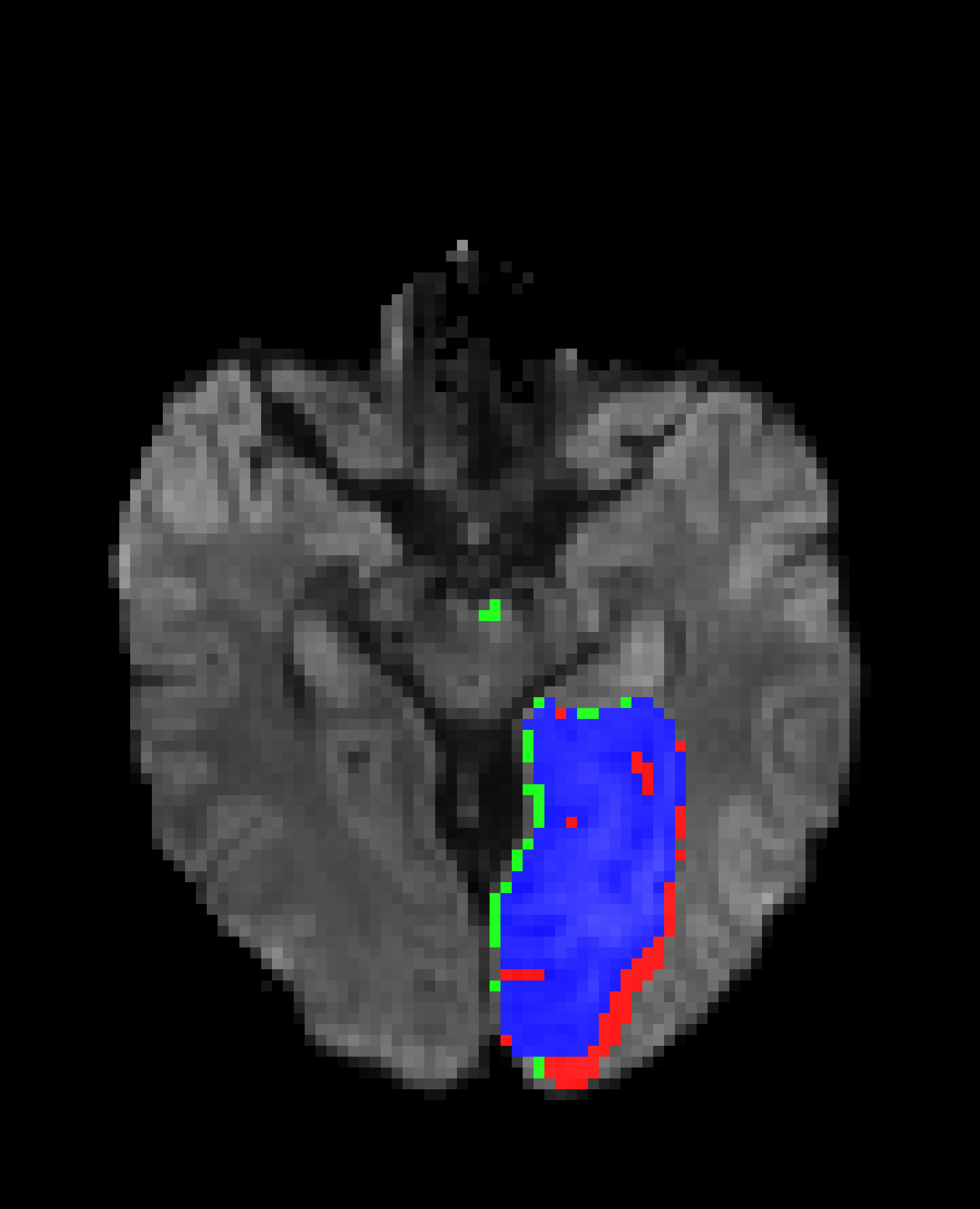}{72.5}{(1.91,2.3)}{.25}
    \end{subfigure}%
    \begin{subfigure}{0.2\textwidth}
        \includegraphix[.98\textwidth]{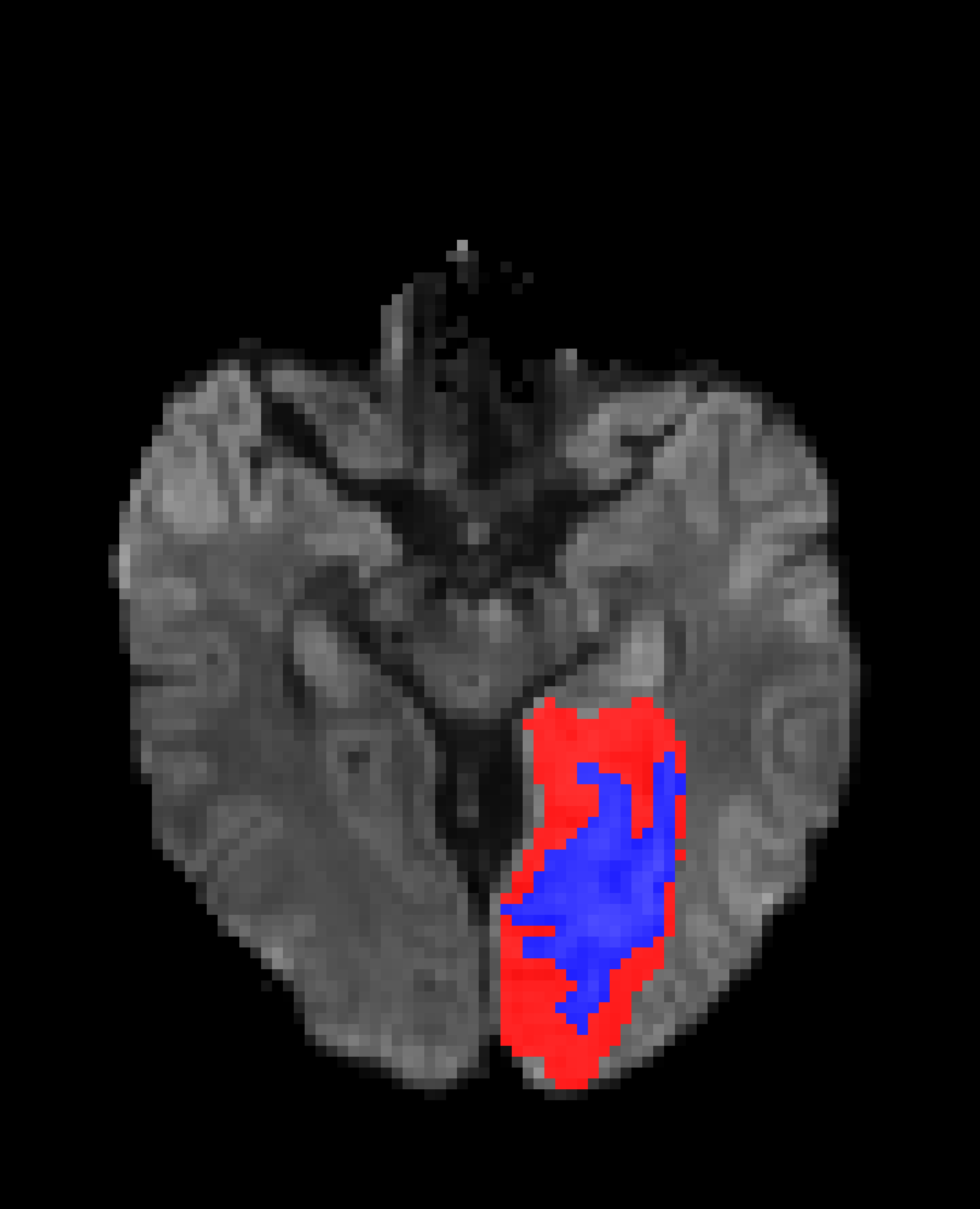}{43.7}
    \end{subfigure}%
    \begin{subfigure}{0.2\textwidth}
        \includegraphix[.98\textwidth]{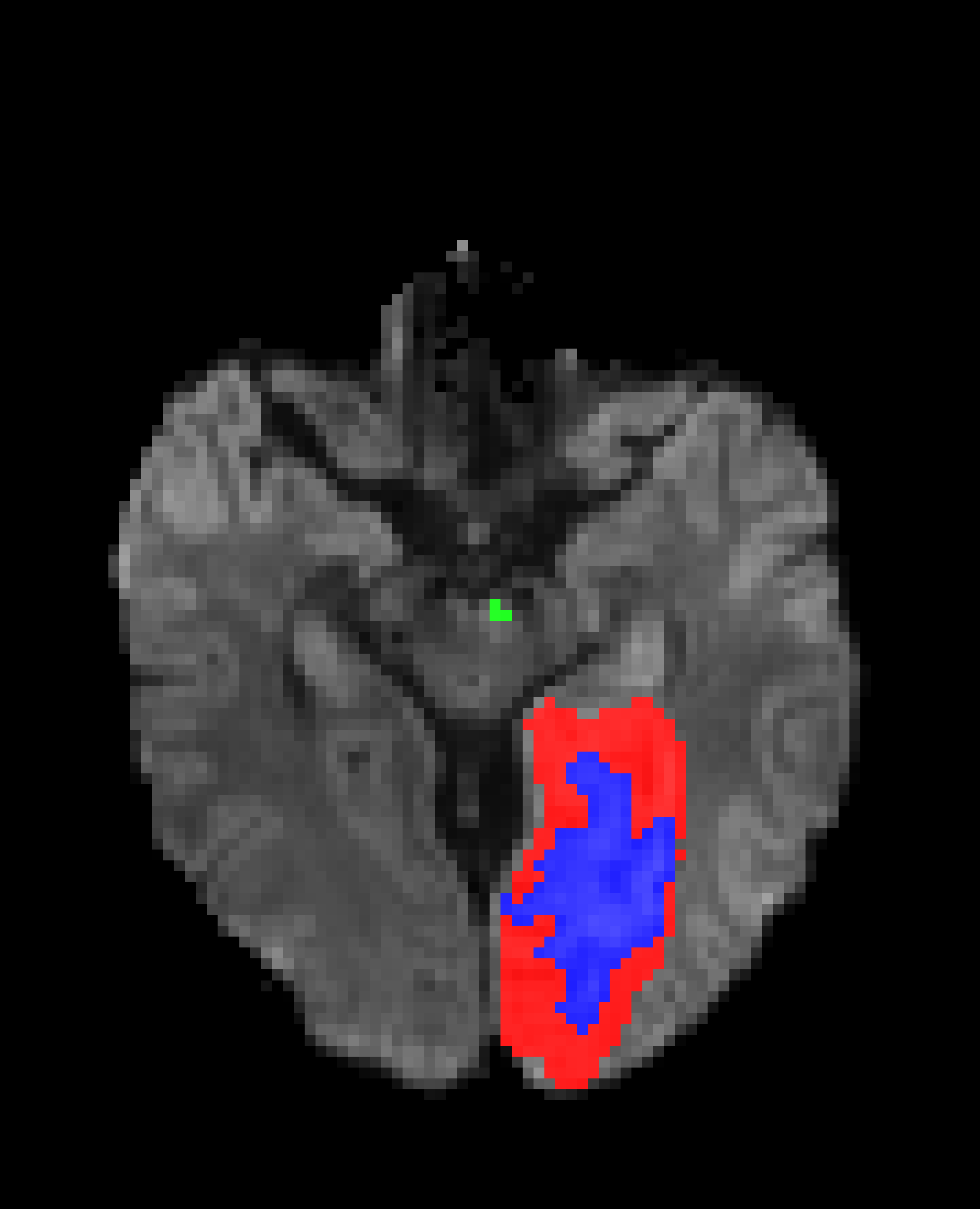}{40.6}{(1.93,2.3)}{.25}
    \end{subfigure}

    \vspace{-5pt}

    \begin{subfigure}{0.2\textwidth}
        \includegraphix[.98\textwidth]{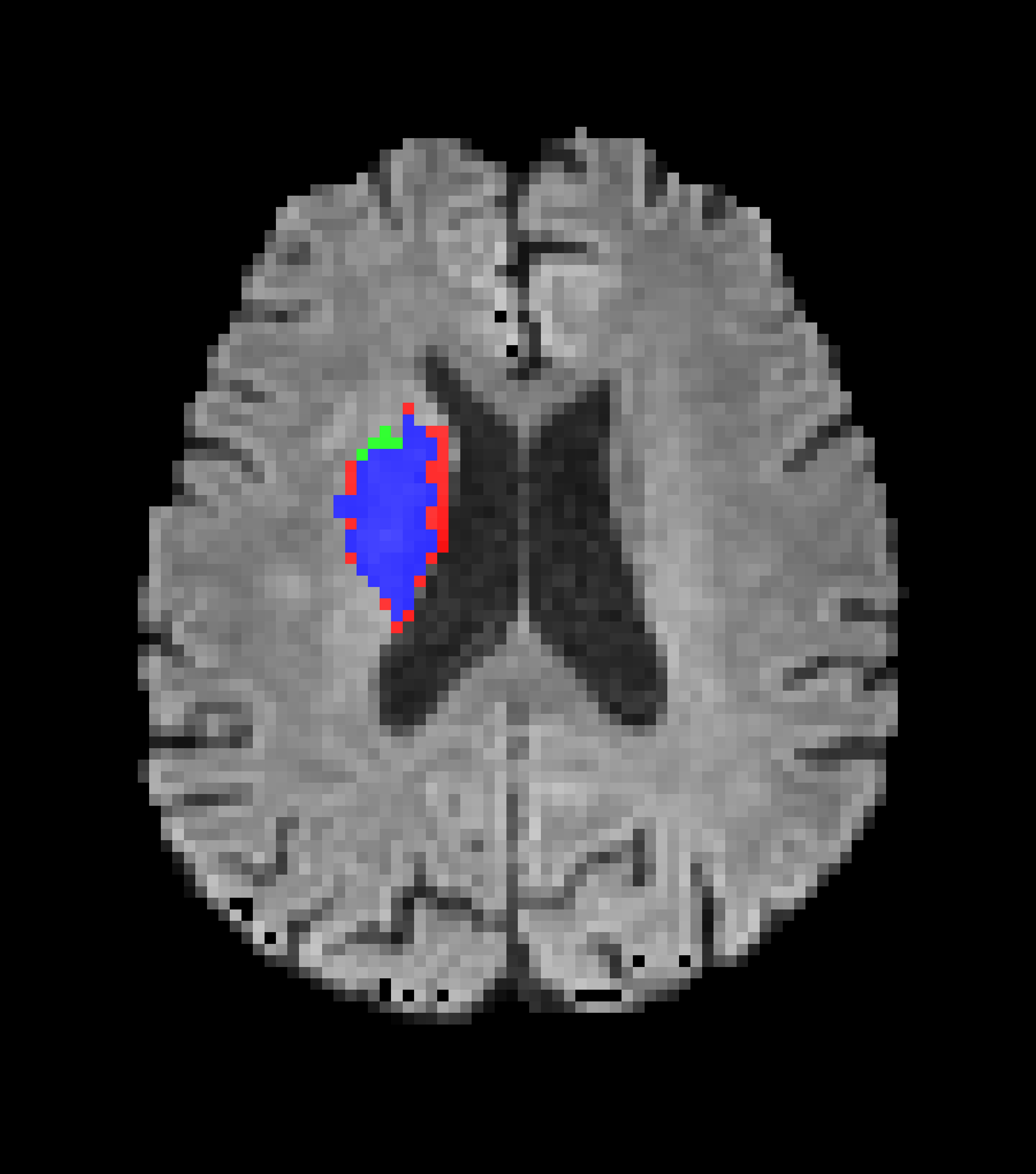}{82.6}
        \caption*{Deconver}
    \end{subfigure}%
    \begin{subfigure}{0.2\textwidth}
        \includegraphix[.98\textwidth]{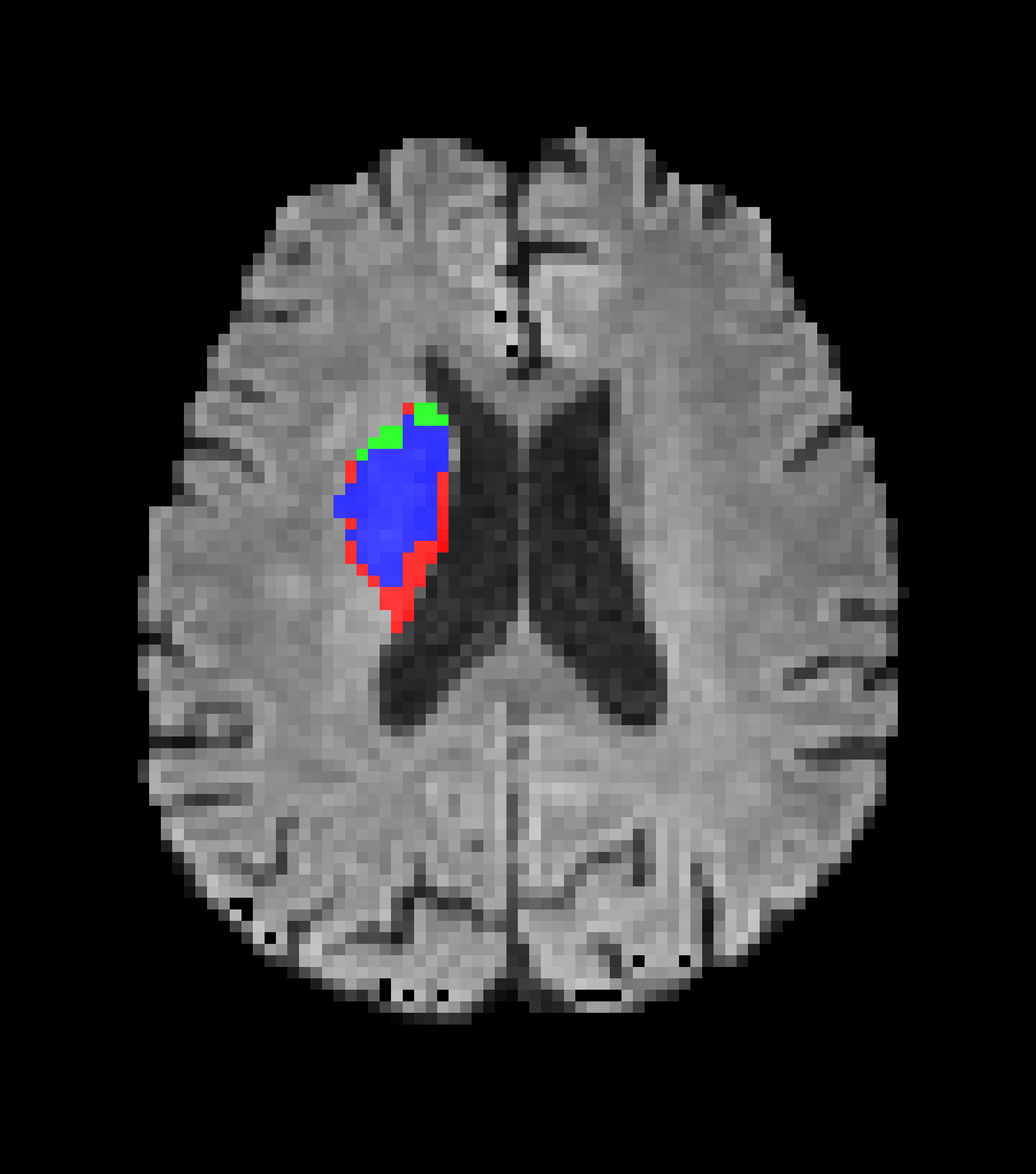}{71.8}
        \caption*{nnU-Net}
    \end{subfigure}%
    \begin{subfigure}{0.2\textwidth}
        \includegraphix[.98\textwidth]{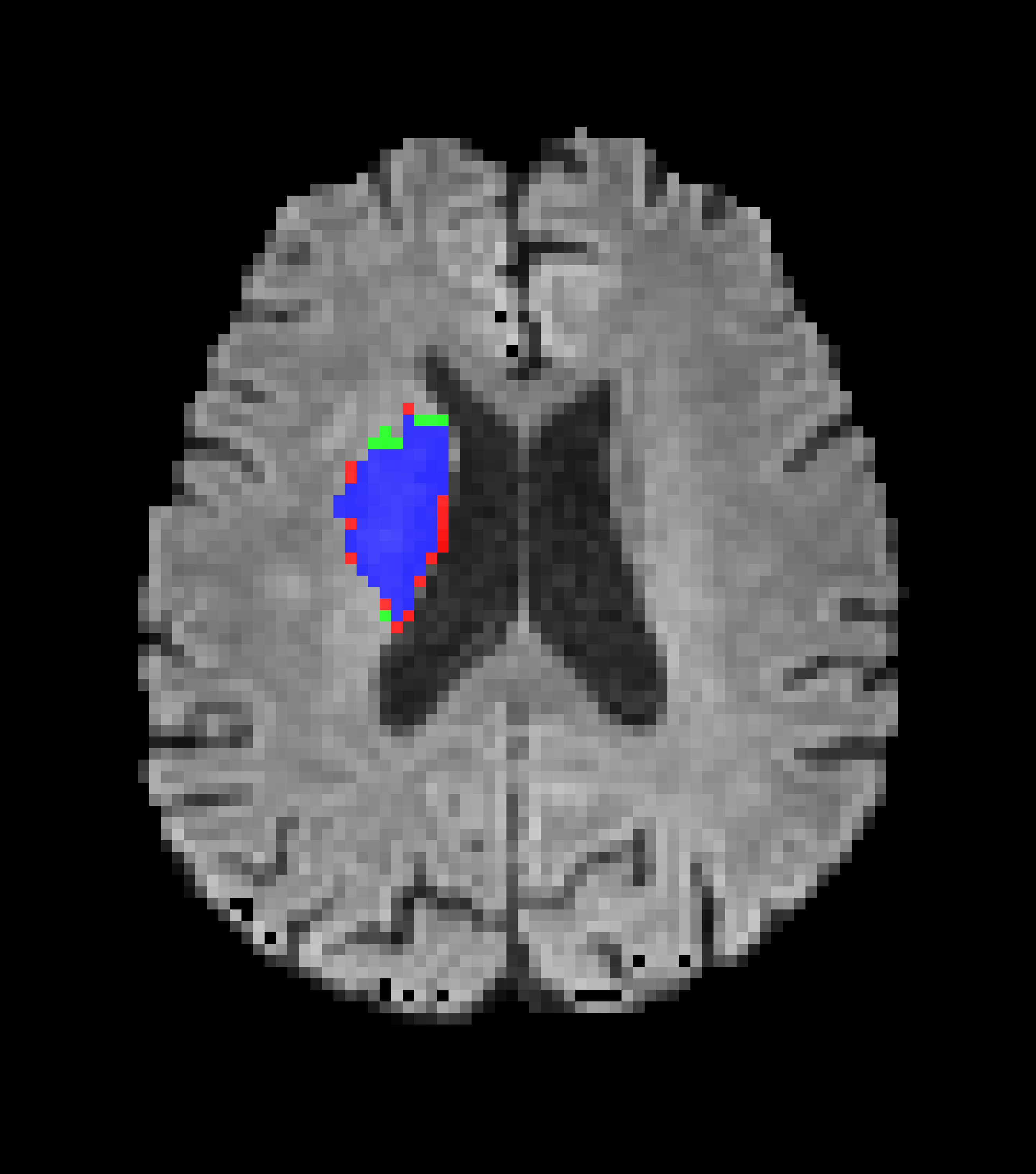}{82.4}
        \caption*{SegResNet}
    \end{subfigure}%
    \begin{subfigure}{0.2\textwidth}
        \includegraphix[.98\textwidth]{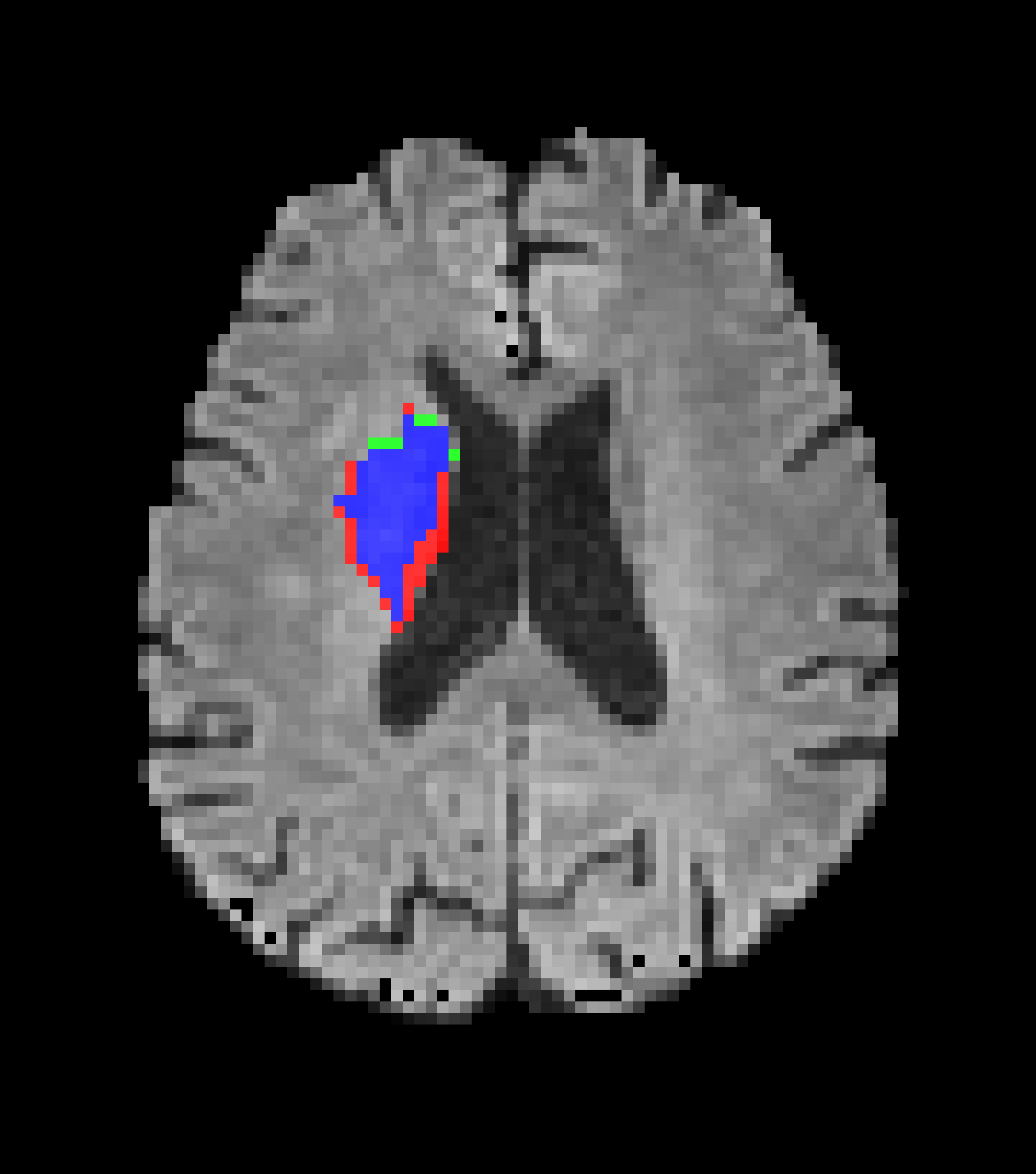}{81.5}
        \caption*{Swin UNETR}
    \end{subfigure}%
    \begin{subfigure}{0.2\textwidth}
        \includegraphix[.98\textwidth]{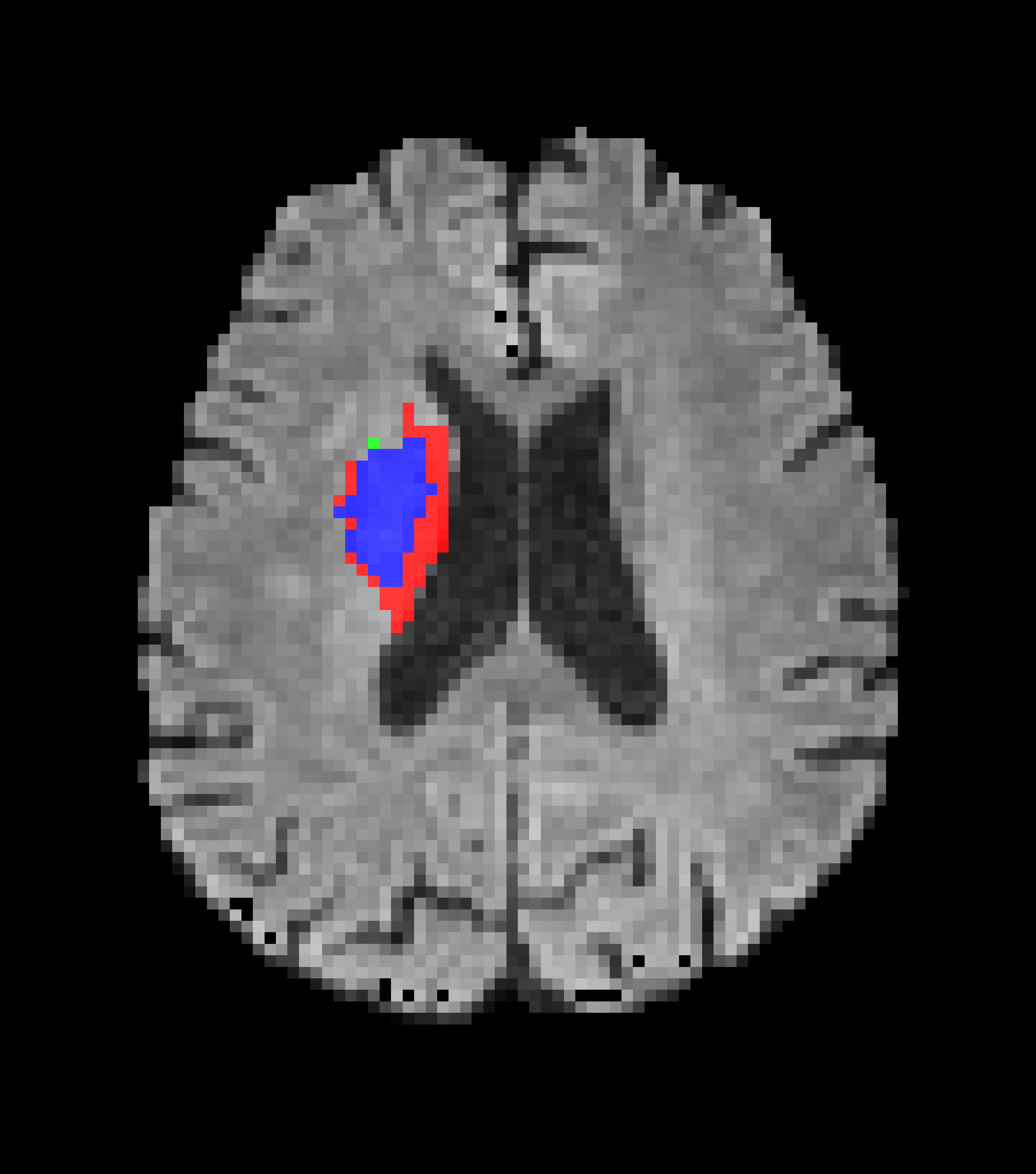}{74.9}
        \caption*{UNETR}
    \end{subfigure}

    \caption{Qualitative results of stroke lesion segmentation on ISLES'22. The regions of true positives are marked in blue, false positives in green, and false negatives in red. The DSC score is presented for each case. SegResNet introduces false positives in the center of the image in the first row (marked with the orange circle), while Swin UNETR, UNETR, and nnU-Net consistently under-segment the lesion in both examples.}
    \label{fig:qualitative_isles}
\end{figure*}
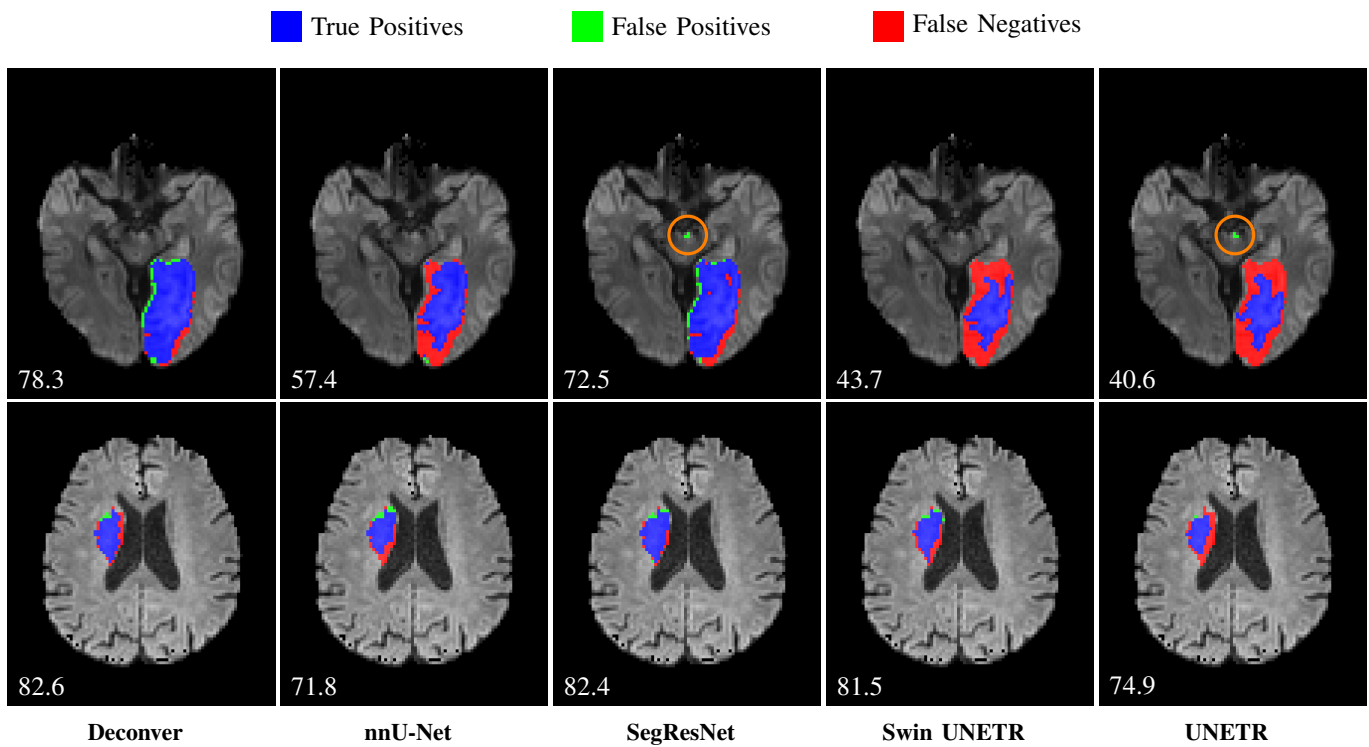

\subsubsection{Brain Tumor Segmentation (BraTS'23)} \label{sec:results_brats}

Table \ref{tab:quantitative_results_brats} presents segmentation performance on BraTS'23, comparing different models across three tumor subregions: enhancing tumor (ET), tumor core (TC), and whole tumor (WT). Both Deconver variants outperform all baselines in average DSC. Notably, Deconver with a kernel size of 3 achieves the highest DSC scores for ET and TC, while Deconver with a kernel size of 5 achieves the top DSC for WT, maintaining strong overall performance.

\begin{table*}[!t]
    \caption{Segmentation performance comparison on BraTS'23. Best results are \textbf{bold}, second-best are \underline{underlined}.}
    \label{tab:quantitative_results_brats}
    \centering
    \renewcommand{\arraystretch}{1.2}
    \resizebox{\textwidth}{!}{
        \begin{tabular}{l | c | c | c c c c | c c c c}
            \toprule
            \multirow{2}{*}{\textbf{Model}} & \multirow{2}{*}{\textbf{Params}} & \multirow{2}{*}{\textbf{FLOPs / voxel}} & \multicolumn{4}{c|}{\textbf{DSC (\%)}} & \multicolumn{4}{c}{\textbf{HD95}}                                                                                                                     \\
            \cline{4-11}
                                            &                                  &                                         & \textbf{ET}                            & \textbf{TC}                       & \textbf{WT}       & \textbf{Avg.}     & \textbf{ET}      & \textbf{TC}      & \textbf{WT}      & \textbf{Avg.}    \\
            \midrule
            nnU-Net                         & 22.6M                            & 921.0K                                  & 86.73                                  & 91.40                             & 93.25             & 90.46             & 3.33             & \underline{3.88} & 5.54             & \underline{4.34} \\
            SegResNet                       & 75.9M                            & 2235.7K                                 & 86.68                                  & 91.49                             & \underline{93.44} & 90.53             & \textbf{3.26}    & 3.92             & \textbf{5.21}    & \textbf{4.17}    \\
            UNETR                           & 139.8M                           & 536.8K                                  & 85.47                                  & 89.92                             & 92.64             & 89.34             & 4.22             & 4.98             & 6.93             & 5.47             \\
            Swin UNETR                      & 62.2M                            & 4098.6K                                 & 86.71                                  & 91.27                             & 93.41             & 90.46             & 3.42             & 3.94             & \underline{5.42} & 4.36             \\
            Factorizer                      & 7.6M                             & 1087.5K                                 & 85.99                                  & \underline{90.51}                 & 93.13             & 89.88             & 3.69             & 4.36             & 5.68             & 4.67             \\
            \midrule
            Deconver (3$\times$3$\times$3)  & 10.6M                            & 167.5K                                  & \textbf{87.01}                         & \textbf{91.56}                    & 93.42             & \textbf{90.66}    & \underline{3.30} & \textbf{3.80}    & 5.63             & 4.45             \\
            Deconver (5$\times$5$\times$5)  & 11.0M                            & 167.5K                                  & \underline{86.97}                      & 91.41                             & \textbf{93.47}    & \underline{90.62} & 3.50             & 3.99             & 5.59             & 4.49             \\
            \bottomrule
        \end{tabular}
    }
\end{table*}

In terms of HD95, Deconver variants demonstrate comparable or superior performance relative to the leading, particularly excelling in TC segmentation. Even when not ranked first, Deconver consistently produces results on par with the best-performing methods.

Similar to the ISLES'22, these competitive segmentation results are achieved with significantly reduced computational complexity. Compared to the second-best performing baseline (SegResNet), Deconver uses approximately 85\% fewer parameters and over 90\% fewer FLOPs.

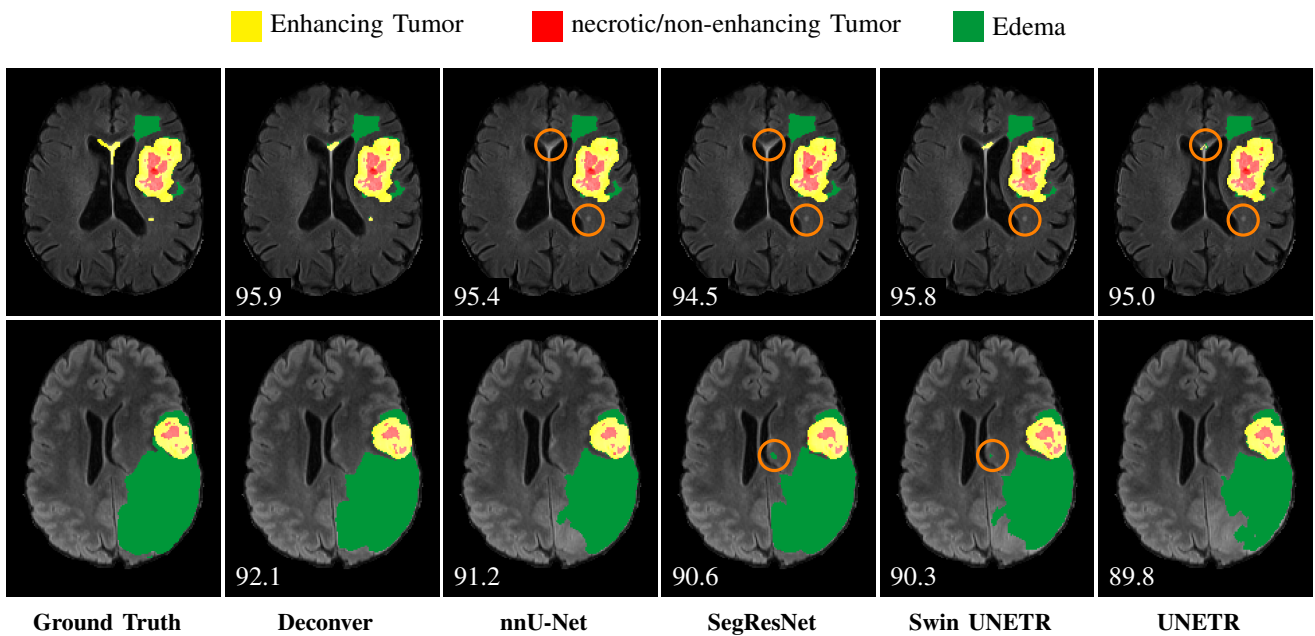
\begin{figure*}[!t]
    \captionsetup[subfigure]{aboveskip=0.25em, belowskip=0.5em, font={bf,small}}
    \centering

    \definecolor{customgreen}{RGB}{0, 151, 57}
    \begin{tikzpicture}
        \fill[yellow] (0,0) rectangle (0.4,0.4);
        \node[right] at (0.4,0.2) {Enhancing Tumor};

        \fill[red] (4,0) rectangle (4.4,0.4);
        \node[right] at (4.4,0.2) {necrotic/non-enhancing Tumor};

        \fill[customgreen] (9.6,0) rectangle (10,0.4);
        \node[right] at (10,0.2) {Edema};
    \end{tikzpicture}

    \vspace{5pt}

    \begin{subfigure}{0.16\textwidth}
        \includegraphix[.98\textwidth]{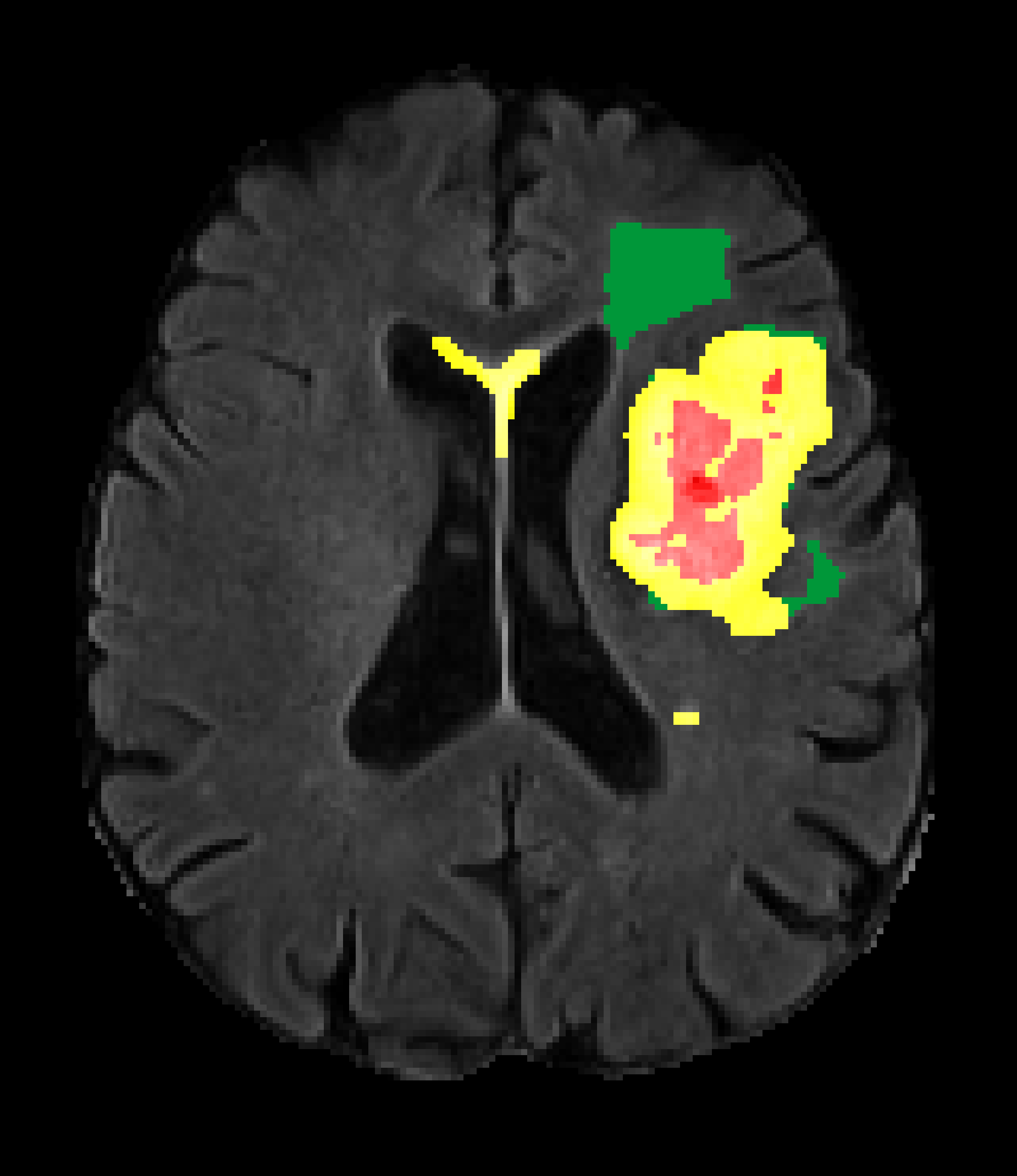}
    \end{subfigure}%
    \begin{subfigure}{0.16\textwidth}
        \includegraphix[.98\textwidth]{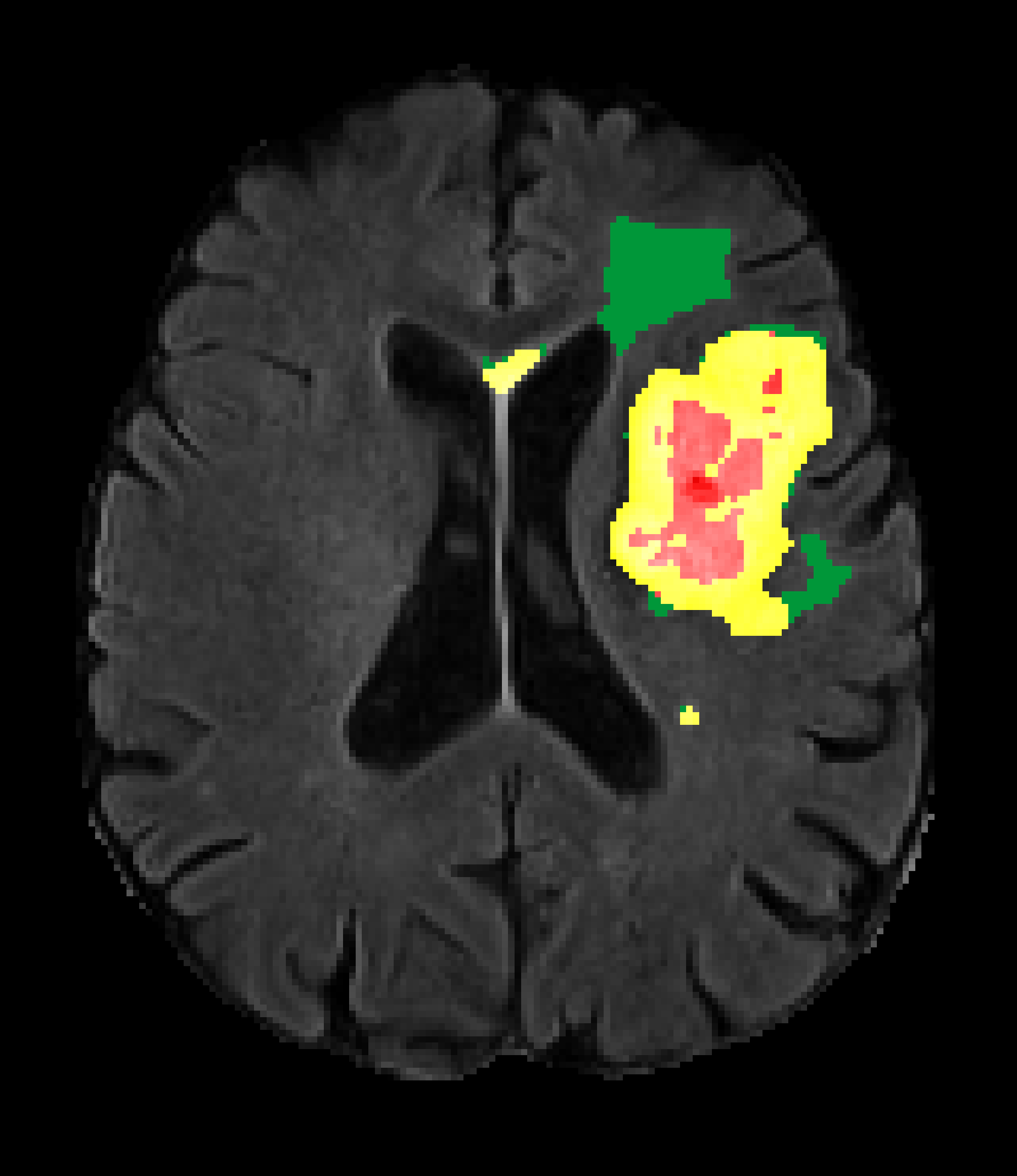}{95.9}
    \end{subfigure}%
    \begin{subfigure}{0.16\textwidth}
        \includegraphix[.98\textwidth]{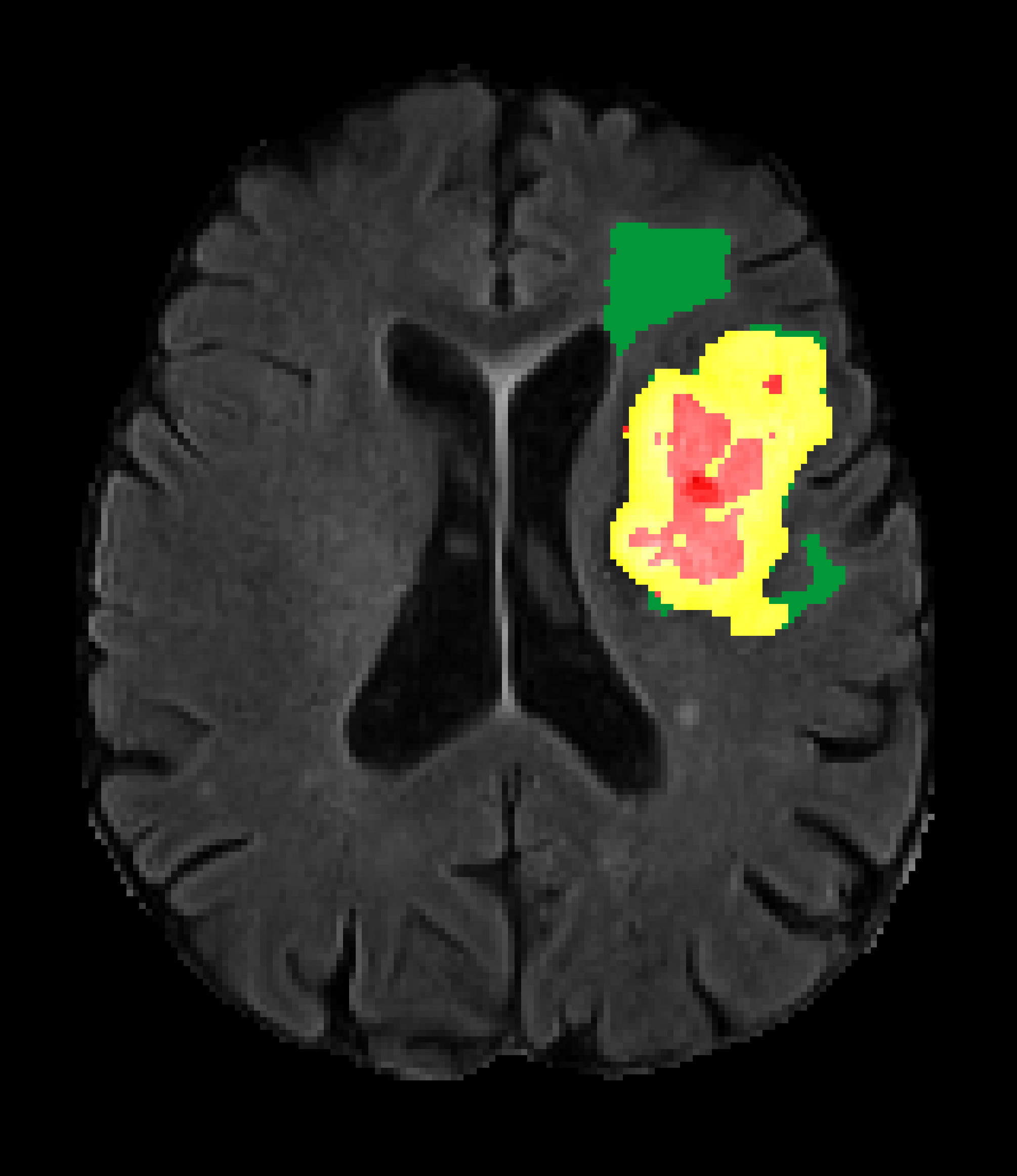}{95.4}{(1.55,2.4)}{.2}{(2.05,1.4)}{.2}
    \end{subfigure}%
    \begin{subfigure}{0.16\textwidth}
        \includegraphix[.98\textwidth]{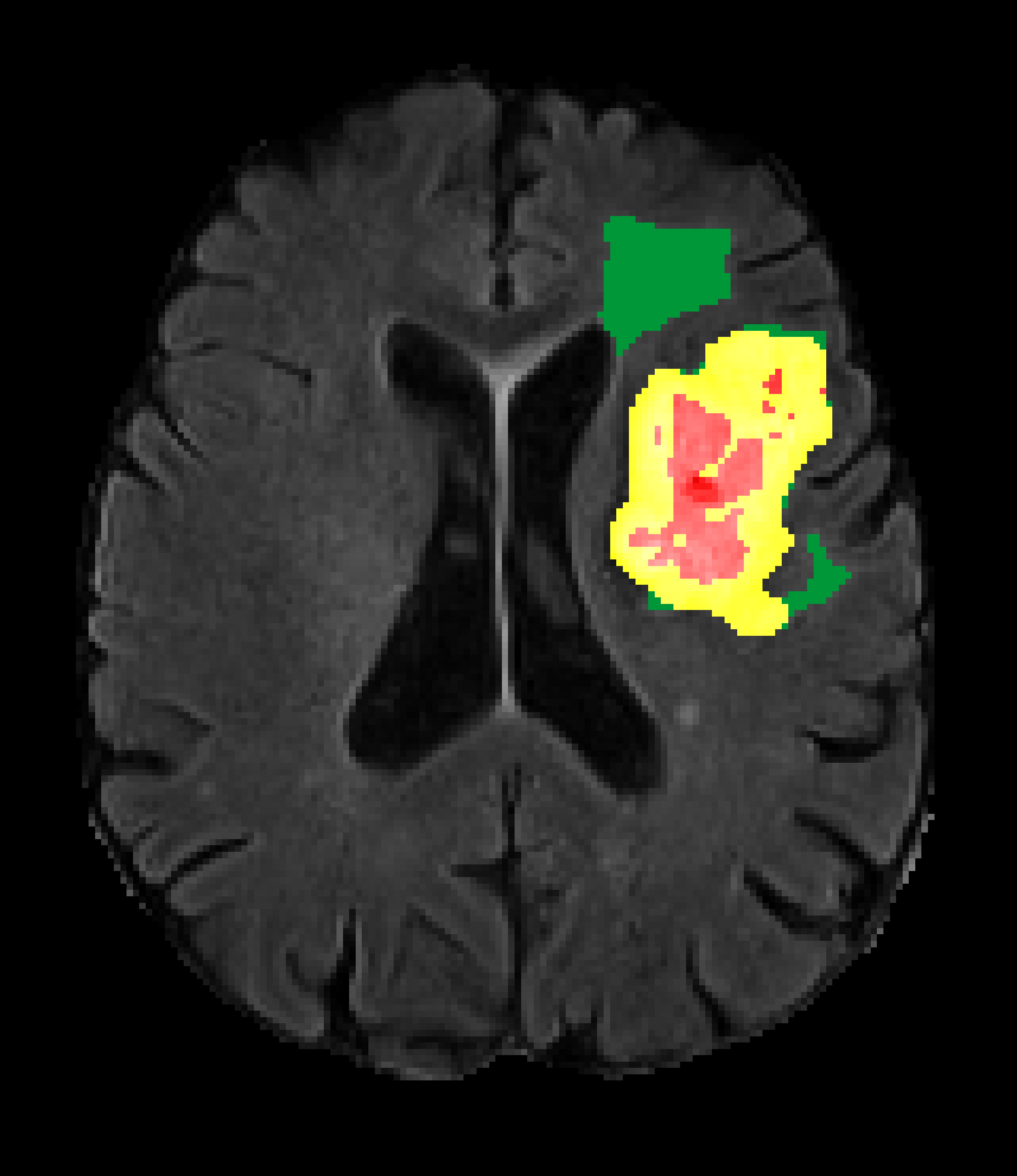}{94.5}{(1.55,2.4)}{.2}{(2.05,1.4)}{.2}
    \end{subfigure}%
    \begin{subfigure}{0.16\textwidth}
        \includegraphix[.98\textwidth]{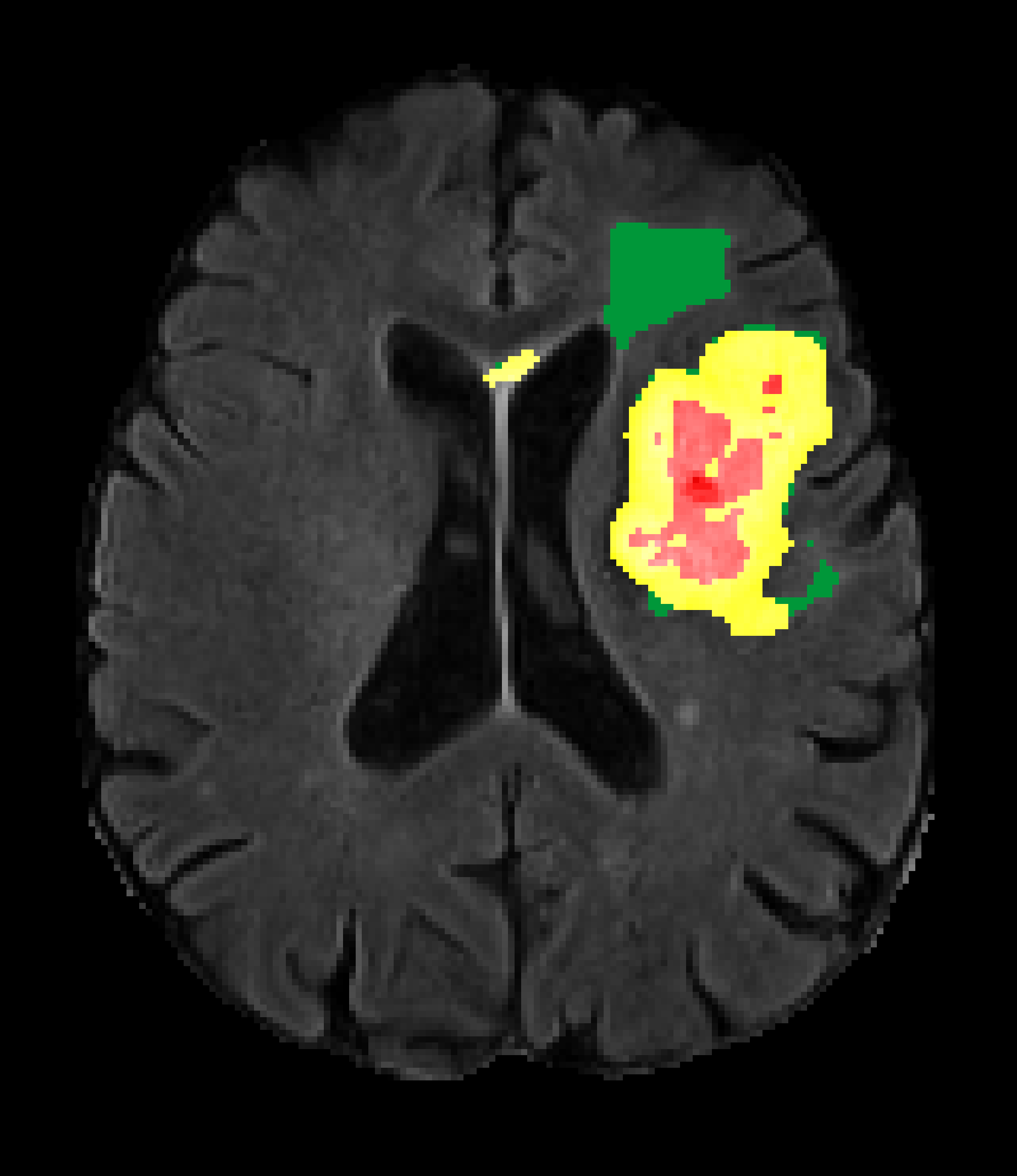}{95.8}{(2.05,1.4)}{.2}
    \end{subfigure}%
    \begin{subfigure}{0.16\textwidth}
        \includegraphix[.98\textwidth]{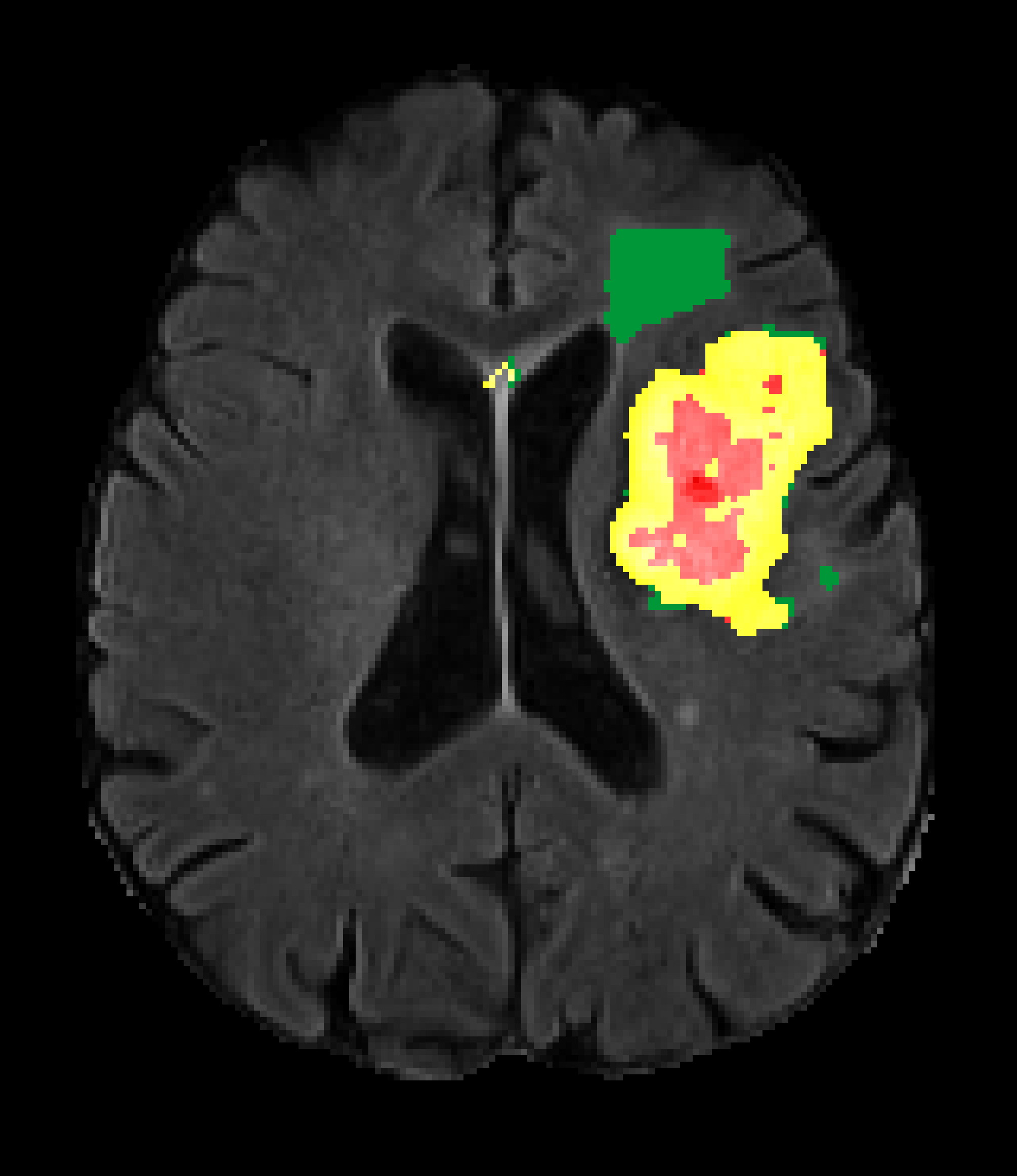}{95.0}{(2.05,1.4)}{.2}{(1.55,2.4)}{.2}
    \end{subfigure}

    \vspace{-5pt}

    \begin{subfigure}{0.16\textwidth}
        \includegraphix[.98\textwidth]{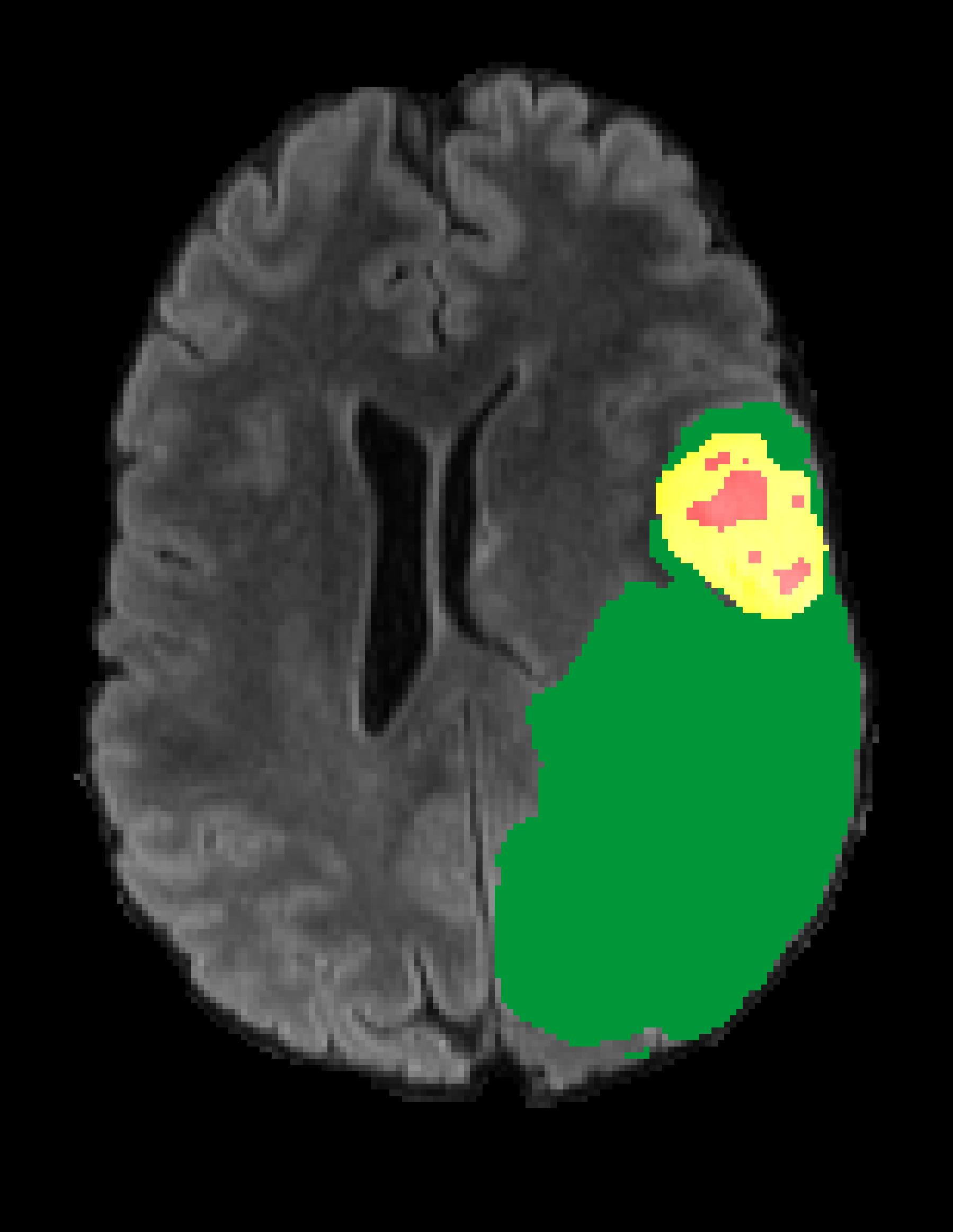}
        \caption*{Ground Truth}
    \end{subfigure}%
    \begin{subfigure}{0.16\textwidth}
        \includegraphix[.98\textwidth]{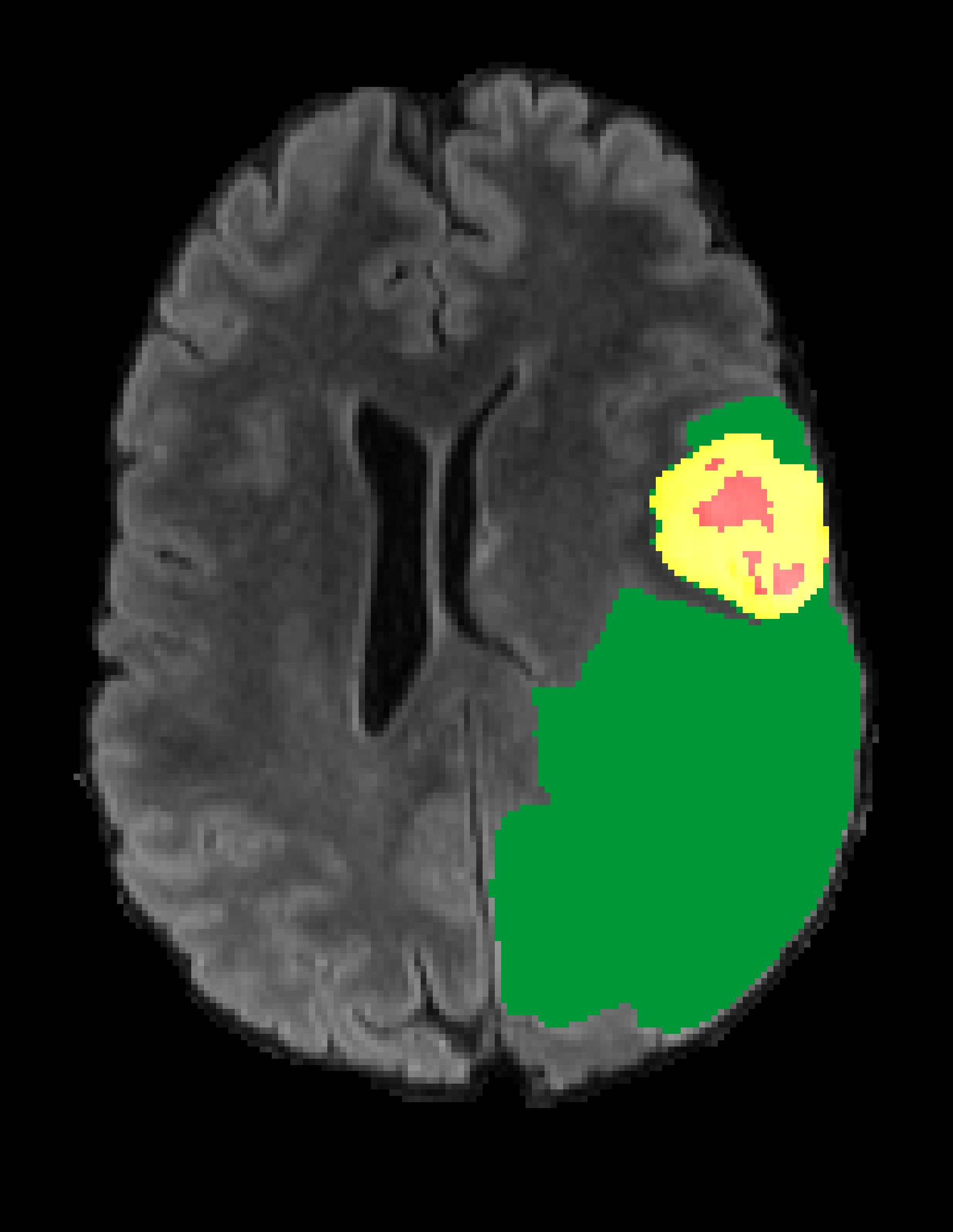}{92.1}
        \caption*{Deconver}
    \end{subfigure}%
    \begin{subfigure}{0.16\textwidth}
        \includegraphix[.98\textwidth]{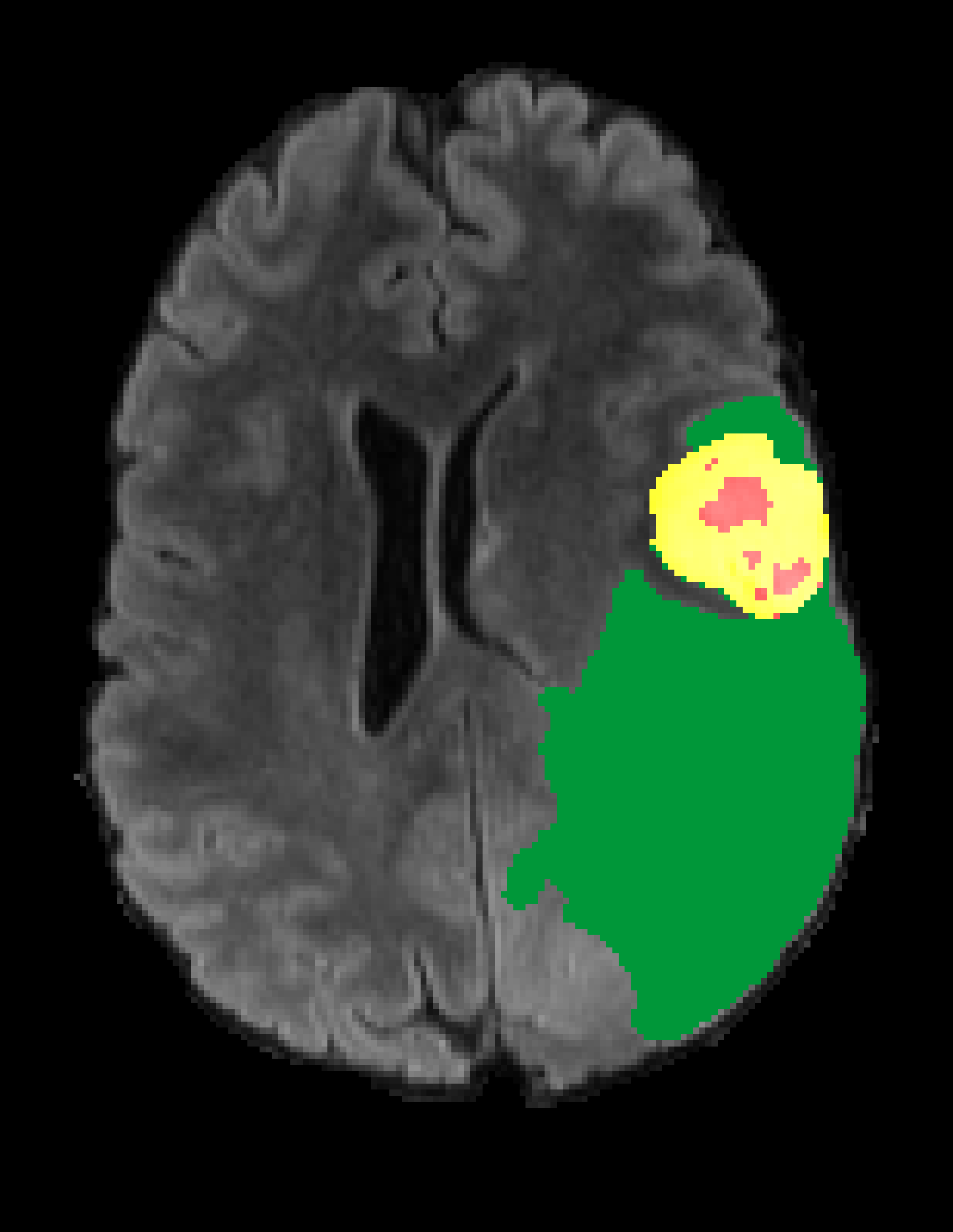}{91.2}
        \caption*{nnU-Net}
    \end{subfigure}%
    \begin{subfigure}{0.16\textwidth}
        \includegraphix[.98\textwidth]{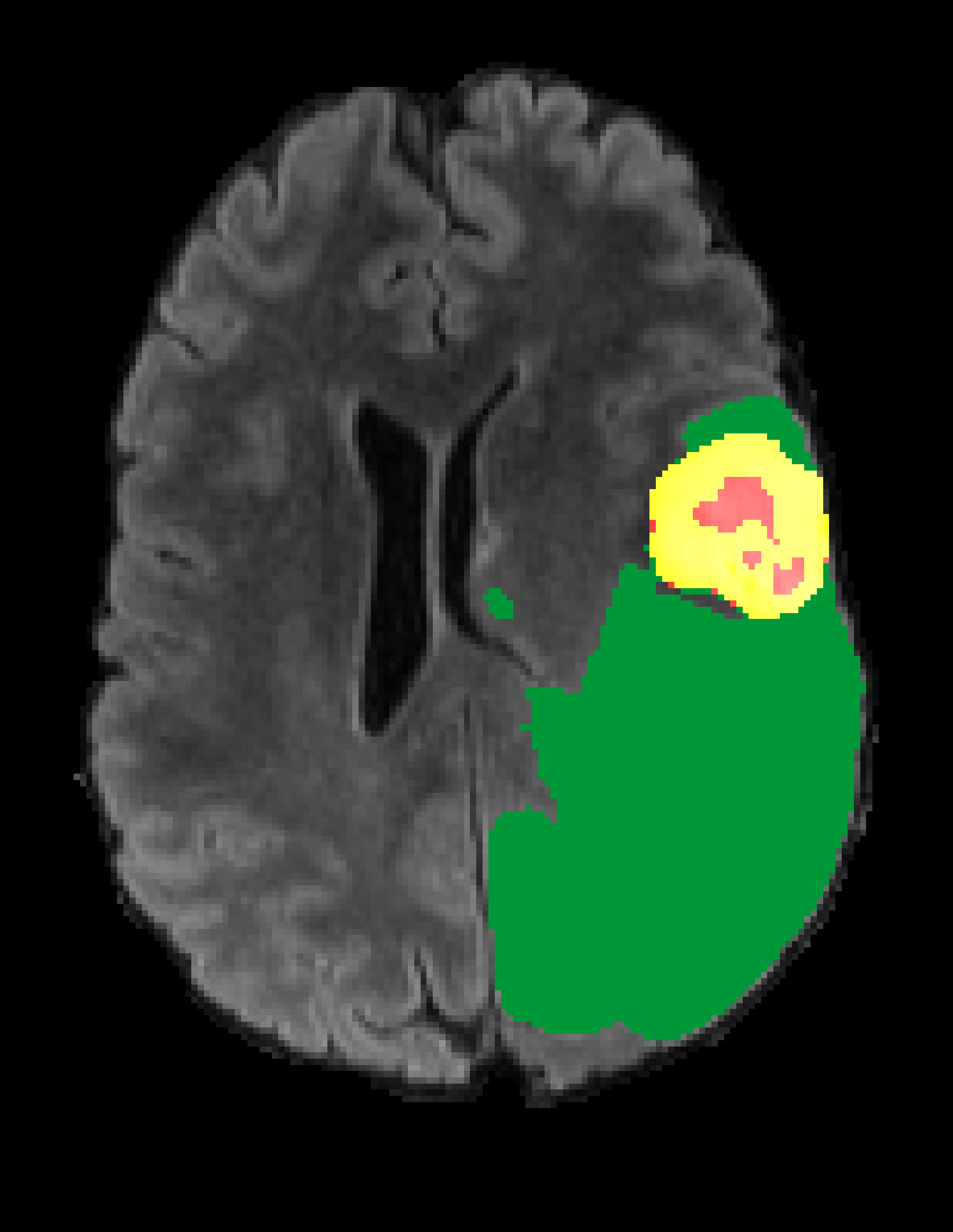}{90.6}{(1.62,2)}{.2}
        \caption*{SegResNet}
    \end{subfigure}%
    \begin{subfigure}{0.16\textwidth}
        \includegraphix[.98\textwidth]{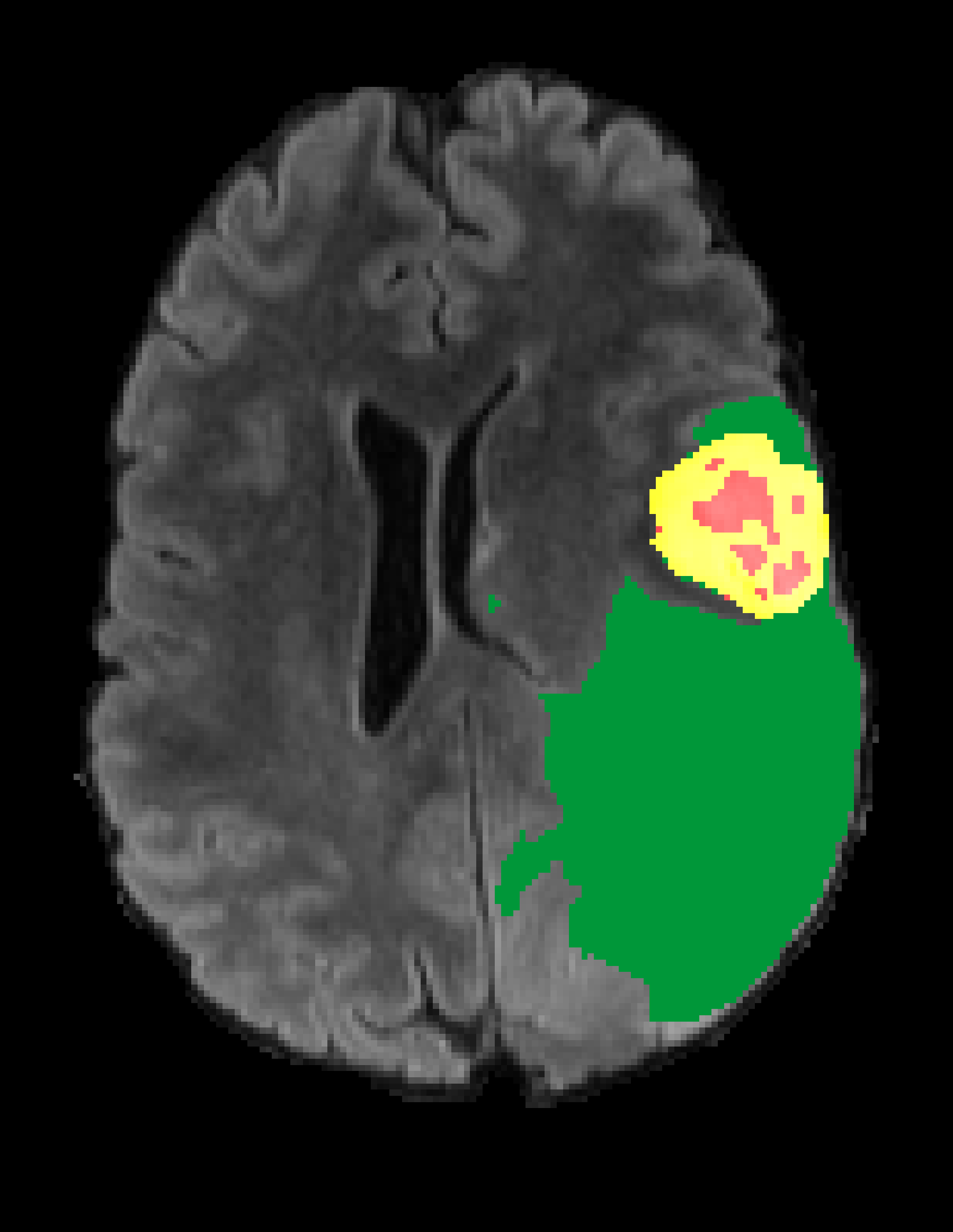}{90.3}{(1.62,2)}{.2}
        \caption*{Swin UNETR}
    \end{subfigure}%
    \begin{subfigure}{0.16\textwidth}
        \includegraphix[.98\textwidth]{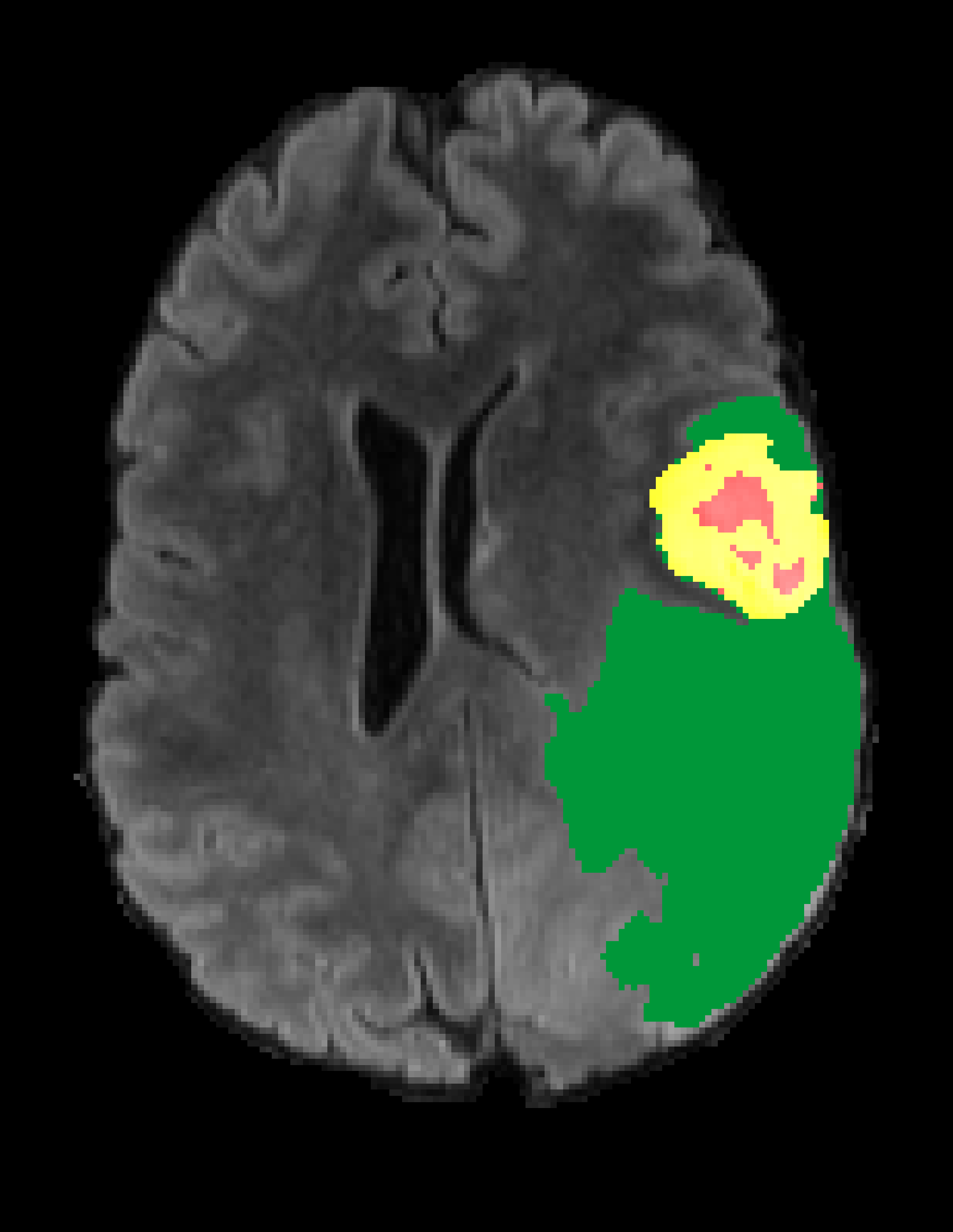}{89.8}
        \caption*{UNETR}
    \end{subfigure}

    \caption{Qualitative results of brain tumor segmentation on BraTS'23. Tumor core (TC) is the union of red (NCR/NET) and yellow (ET) regions, and whole tumor (WT) is the union of green (edema), red, and yellow regions. Each row displays a sample slice from a subject in the validation set. The average DSC score is presented for each case. In the first row example, all of the baselines fail to detect part of the enhancing tumor marked by the orange circle. In the second row example, nnU-Net, Swin UNETR and UNETR do not capture fully the edema, while SegResNet and Swin UNETR falsy predict the area marked by the orange circle as edema.}
    \label{fig:qualitative_brats}
\end{figure*}

Fig. \ref{fig:qualitative_brats} provides qualitative segmentation results on the BraTS'23 dataset. In the first row example, Deconver successfully captures both the centrally located enhancing tumor and the smaller enhancing spot at the lower half, while all other baseline models miss one or both of them (marked by the orange circles). In the second row example, Deconver provides a more accurate delineation of the edema region. Conversely, nnU-Net, Swin UNETR, and UNETR significantly undersegment the edema area, leading to incomplete tumor coverage. SegResNet and Swin UNETR falsely predict the normal brain tissue marked by the orange circle as edema.

\subsection{Results: 2D Segmentation} \label{sec:results_2d}

We further evaluated Deconver on 2D medical image segmentation tasks using the GlaS and FIVES datasets. Table \ref{tab:quantitative_results_2d} presents the quantitative results. Deconver (\( 5 \times 5 \)) achieved the highest DSC scores on both datasets, with {92.12\%} on GlaS and 92.72\% on FIVES. Furthermore, it obtained the lowest HD95 value on GlaS (60.49) and the second-best FIVES (30.26). While Deconver (\( 3 \times 3 \)) demonstrated comparable performance relative to baseline methods, it remained outperformed by its \( 5 \times 5 \) counterpart. These results suggest that larger kernel sizes lead to improved performance in 2D medical image segmentation.

\begin{table}[!t]
    \caption{Segmentation performance comparison on 2D datasets (GlaS and FIVES). Best results are \textbf{bold}, second-best are \underline{underlined}.}
    \label{tab:quantitative_results_2d}
    \centering
    \renewcommand{\arraystretch}{1.2}
    \resizebox{\linewidth}{!}{
        \begin{tabular}{l | c | c | c c | c c}
            \toprule
            \multirow{2}{*}{\textbf{Model}} & \multirow{2}{*}{\textbf{Params}} & \multirow{2}{*}{\textbf{FLOPs / pixel}} & \multicolumn{2}{c|}{\textbf{GlaS}} & \multicolumn{2}{c}{\textbf{FIVES}}                                         \\
            \cline{4-7}
                                            &                                  &                                         & \textbf{DSC (\%)}                  & \textbf{HD95}                      & \textbf{DSC (\%)} & \textbf{HD95}     \\
            \midrule
            nnU-Net                         & 20.6M                            & 874.8K                                  & 91.61                              & 87.17                              & 92.65             & 33.01             \\
            SegResNet                       & 25.5M                            & 1928.8K                                 & 91.23                              & 69.64                              & \underline{92.71} & \textbf{28.80}    \\
            UNETR                           & 120.2M                           & 1195.1K                                 & 90.45                              & 73.38                              & 90.98             & 35.40             \\
            Swin UNETR                      & 25.1M                            & 1406.0K                                 & \underline{91.70}                  & \underline{67.27}                  & 92.69             & 30.87             \\
            \midrule
            Deconver (3\(\times\)3)         & 20.6M                            & 422.8K                                  & 91.52                              & 68.52                              & 92.48             & 35.45             \\
            Deconver (5\(\times\)5)         & 20.8M                            & 422.8K                                  & \textbf{92.12}                     & \textbf{60.49}                     & \textbf{92.72}    & \underline{30.26} \\
            \bottomrule
        \end{tabular}
    }
\end{table}

\begin{figure*}[!t]
    \captionsetup[subfigure]{aboveskip=0.25em, belowskip=0.5em, font={bf,small}}
    \centering

    \begin{tikzpicture}
        \fill[blue] (0,0) rectangle (0.4,0.4);
        \node[right] at (0.4,0.2) {True Positives};

        \fill[green] (4,0) rectangle (4.4,0.4);
        \node[right] at (4.4,0.2) {False Positives};

        \fill[red] (8,0) rectangle (8.4,0.4);
        \node[right] at (8.4,0.2) {False Negatives};
    \end{tikzpicture}

    \vspace{5pt}

    \begin{subfigure}{0.19\textwidth}
        \includegraphix[.98\textwidth]{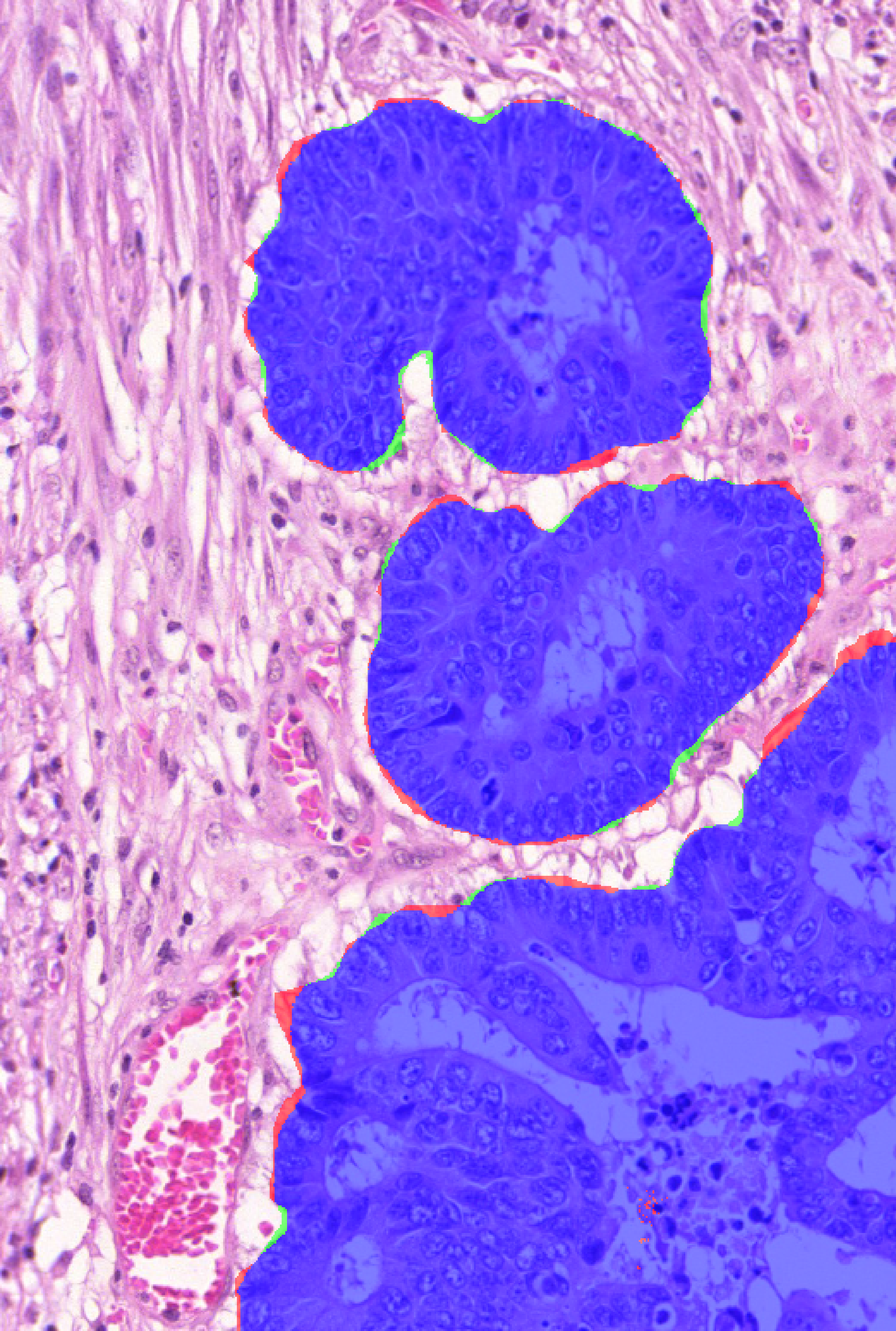}{98.6}
    \end{subfigure}%
    \begin{subfigure}{0.19\textwidth}
        \includegraphix[.98\textwidth]{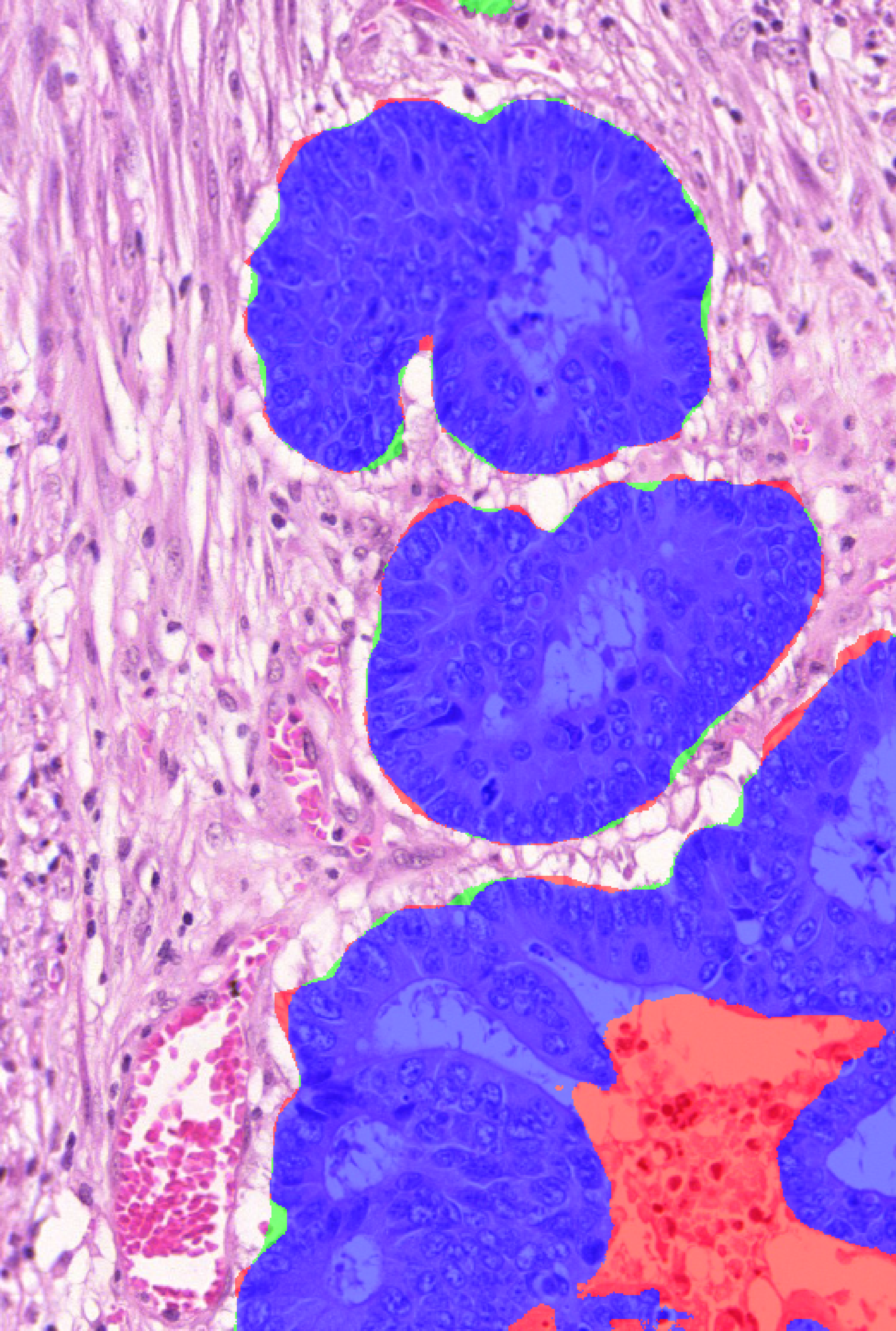}{91.4}
    \end{subfigure}%
    \begin{subfigure}{0.19\textwidth}
        \includegraphix[.98\textwidth]{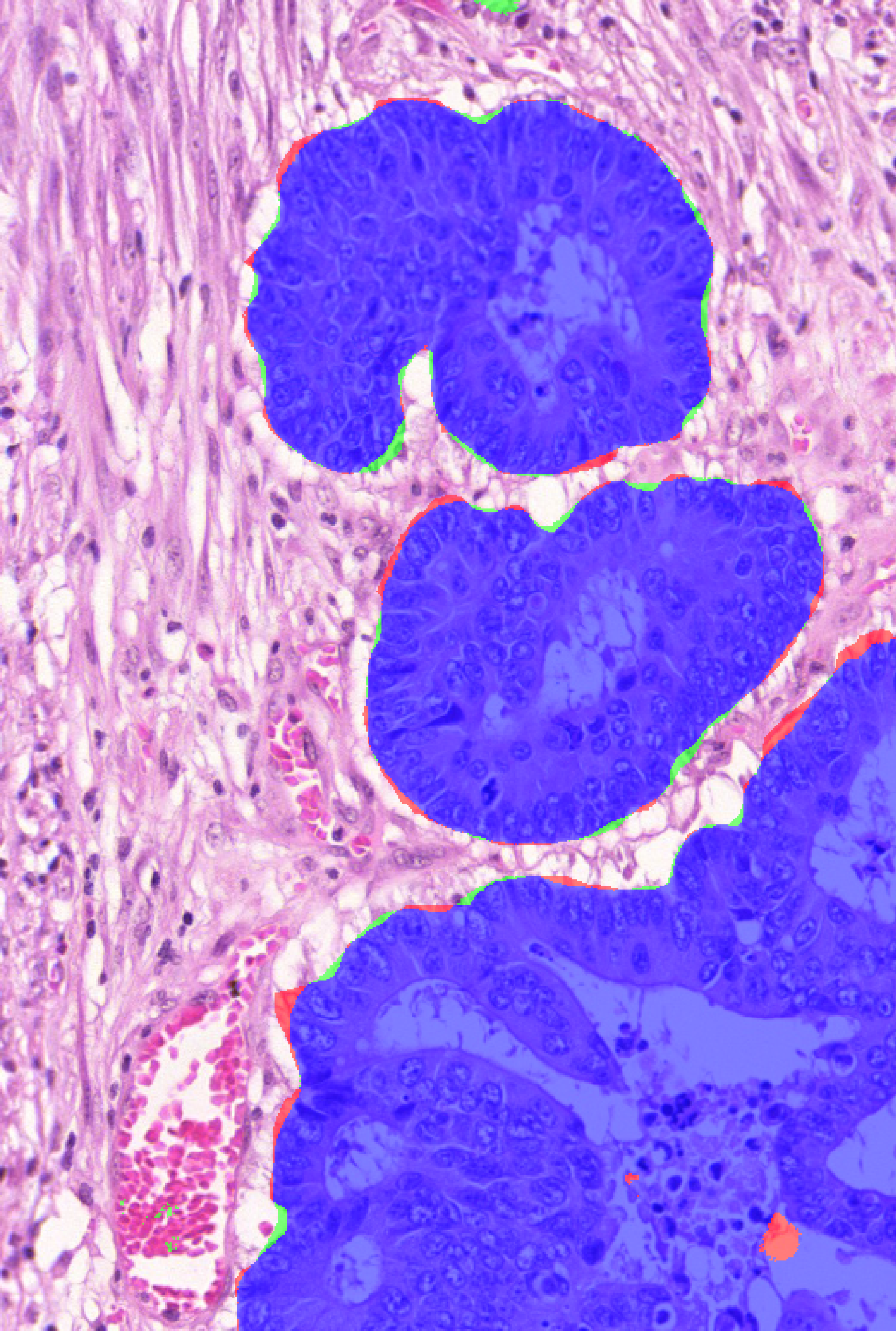}{98.5}{(3.05,0.5)}{.3}
    \end{subfigure}%
    \begin{subfigure}{0.19\textwidth}
        \includegraphix[.98\textwidth]{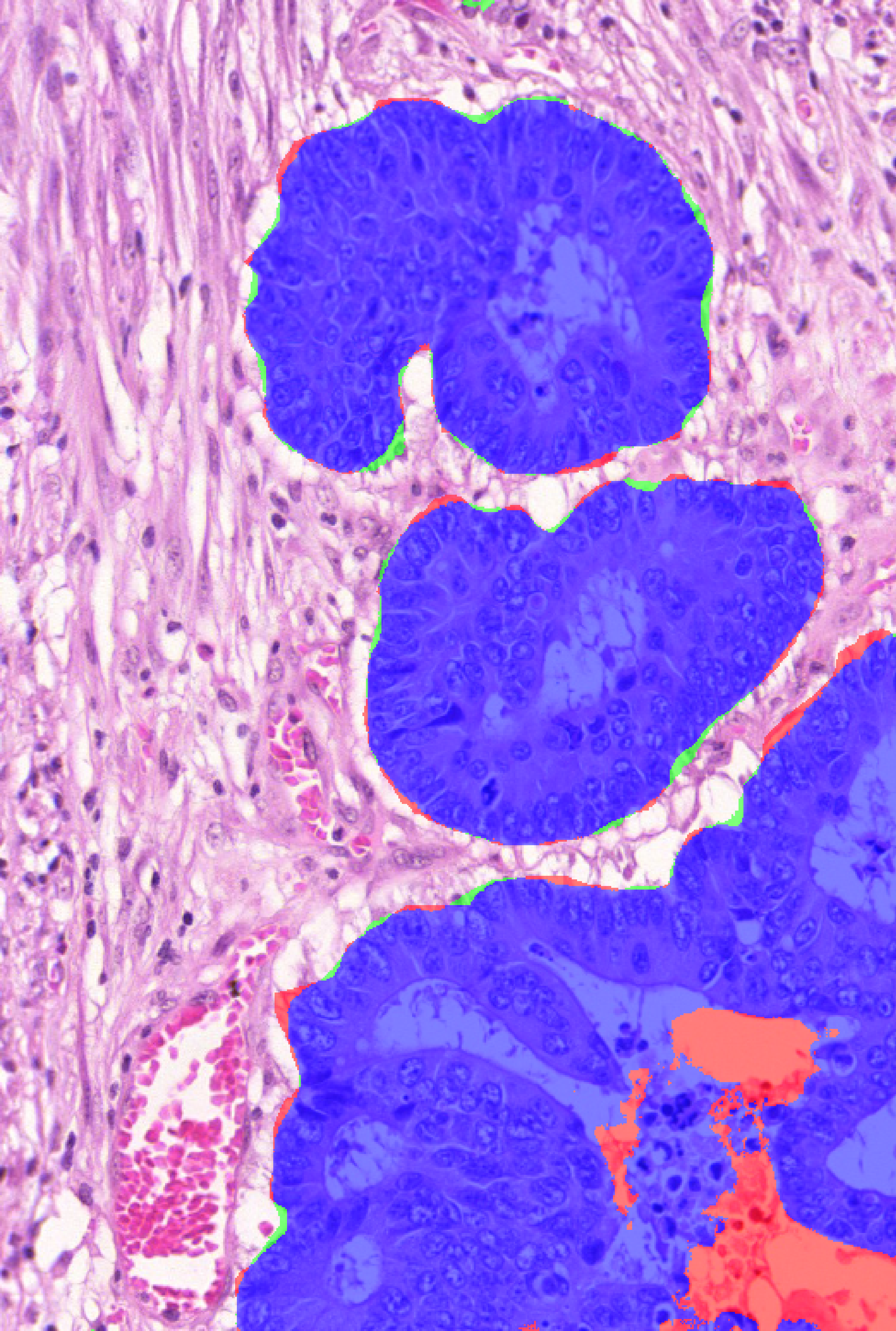}{95.3}
    \end{subfigure}%
    \begin{subfigure}{0.19\textwidth}
        \includegraphix[.98\textwidth]{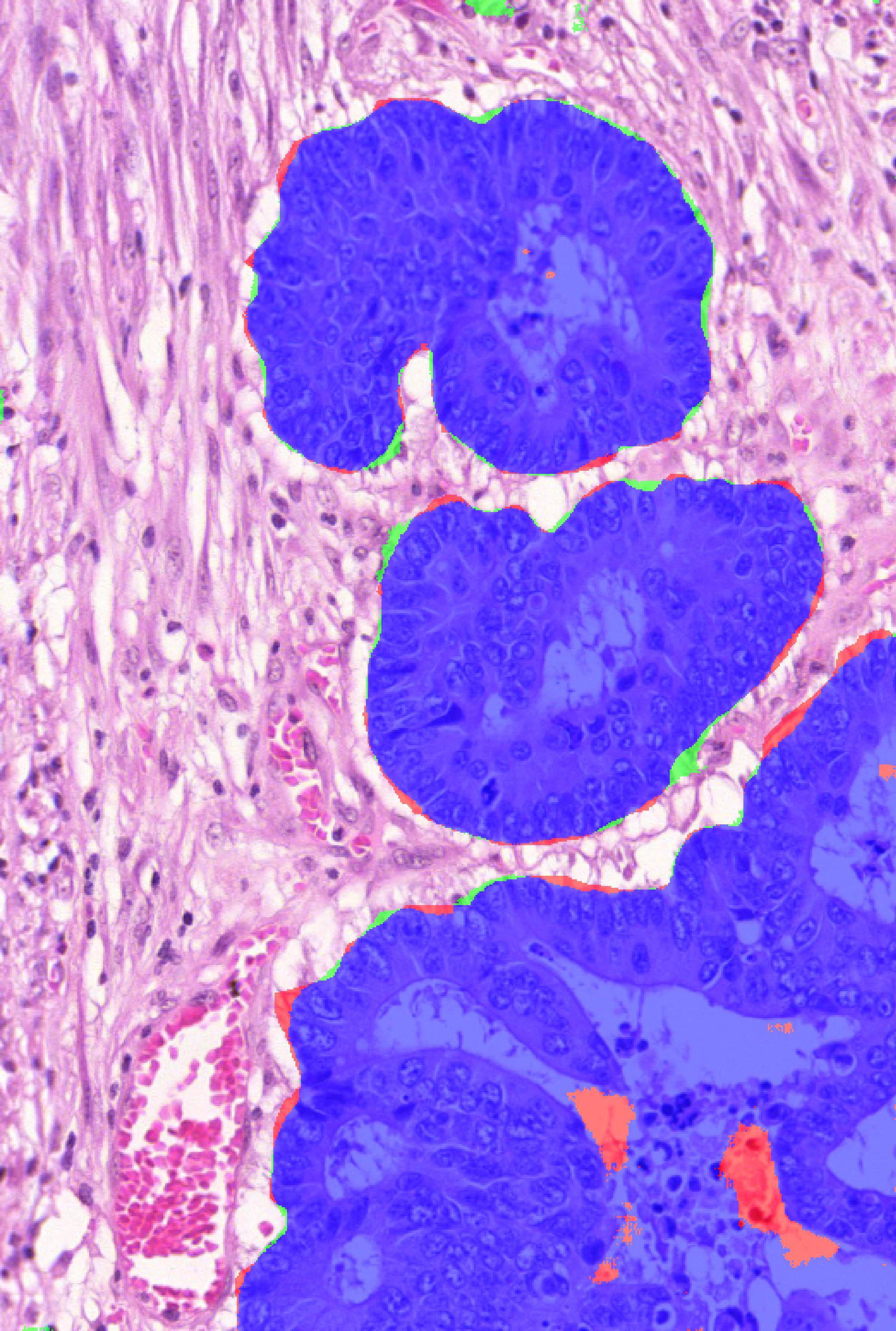}{97.6}
    \end{subfigure}

    \vspace{-5pt}

    \begin{subfigure}{0.19\textwidth}
        \includegraphix[.98\textwidth]{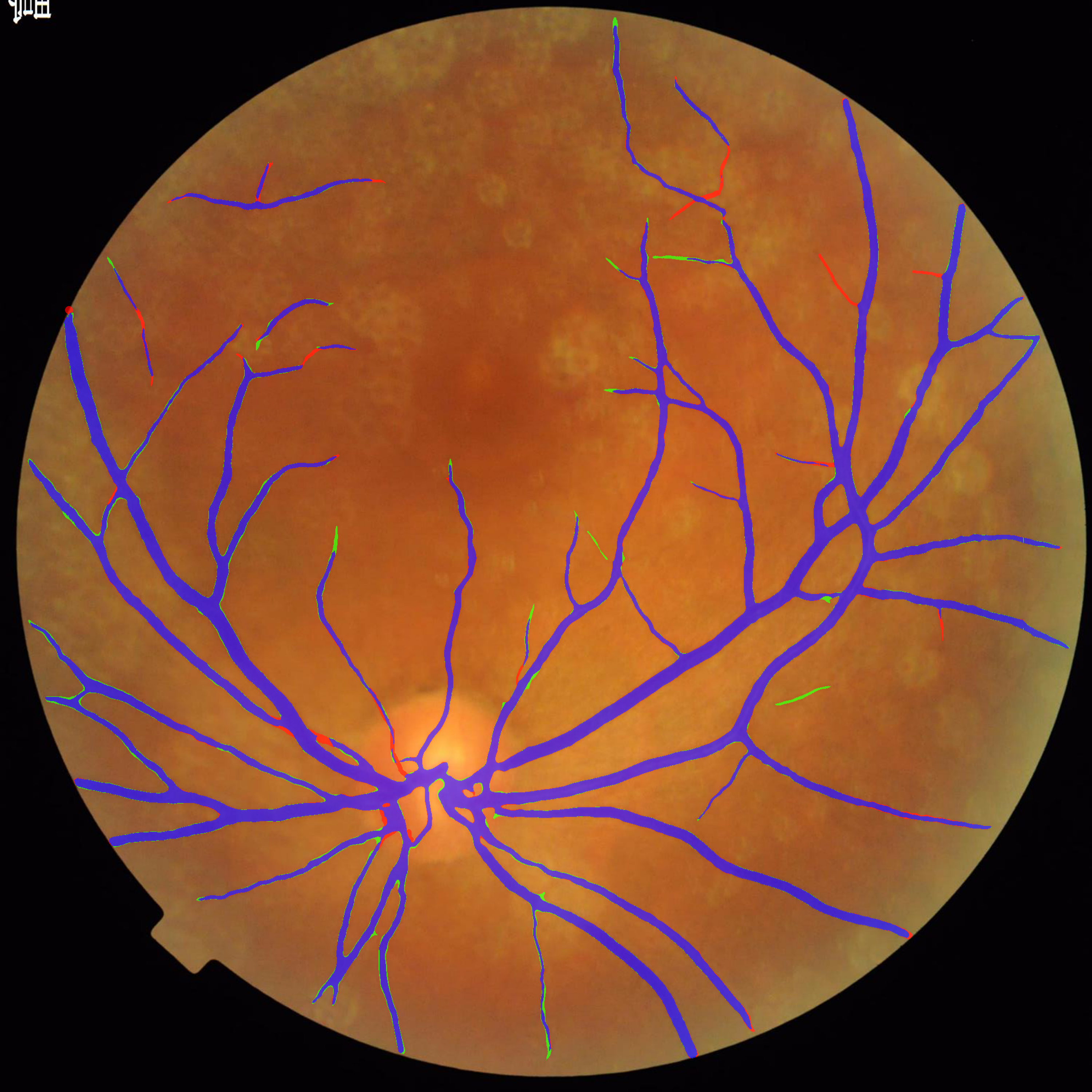}{95.4}
        \caption*{Deconver}
    \end{subfigure}%
    \begin{subfigure}{0.19\textwidth}
        \includegraphix[.98\textwidth]{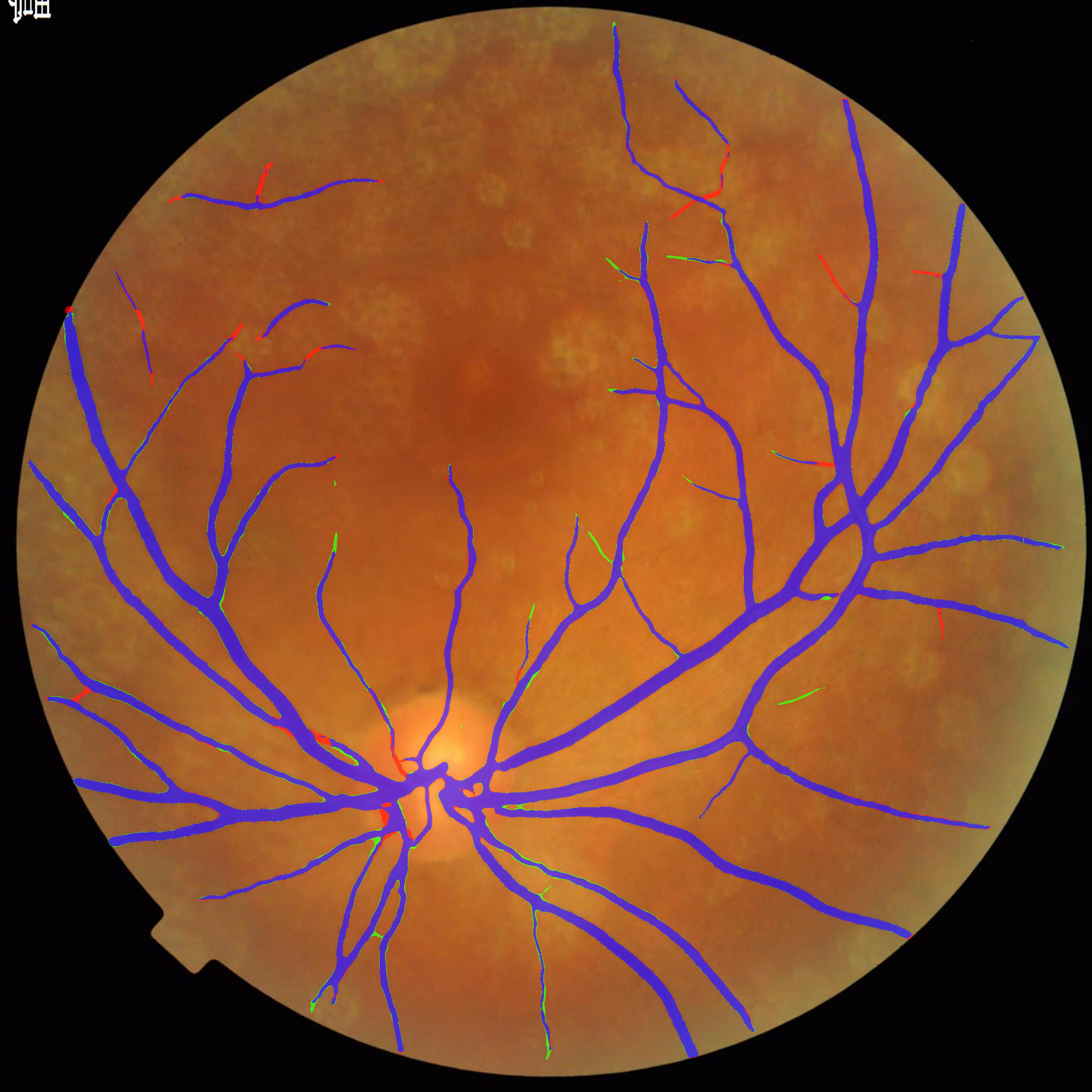}{95.3}
        \caption*{nnU-Net}
    \end{subfigure}%
    \begin{subfigure}{0.19\textwidth}
        \includegraphix[.98\textwidth]{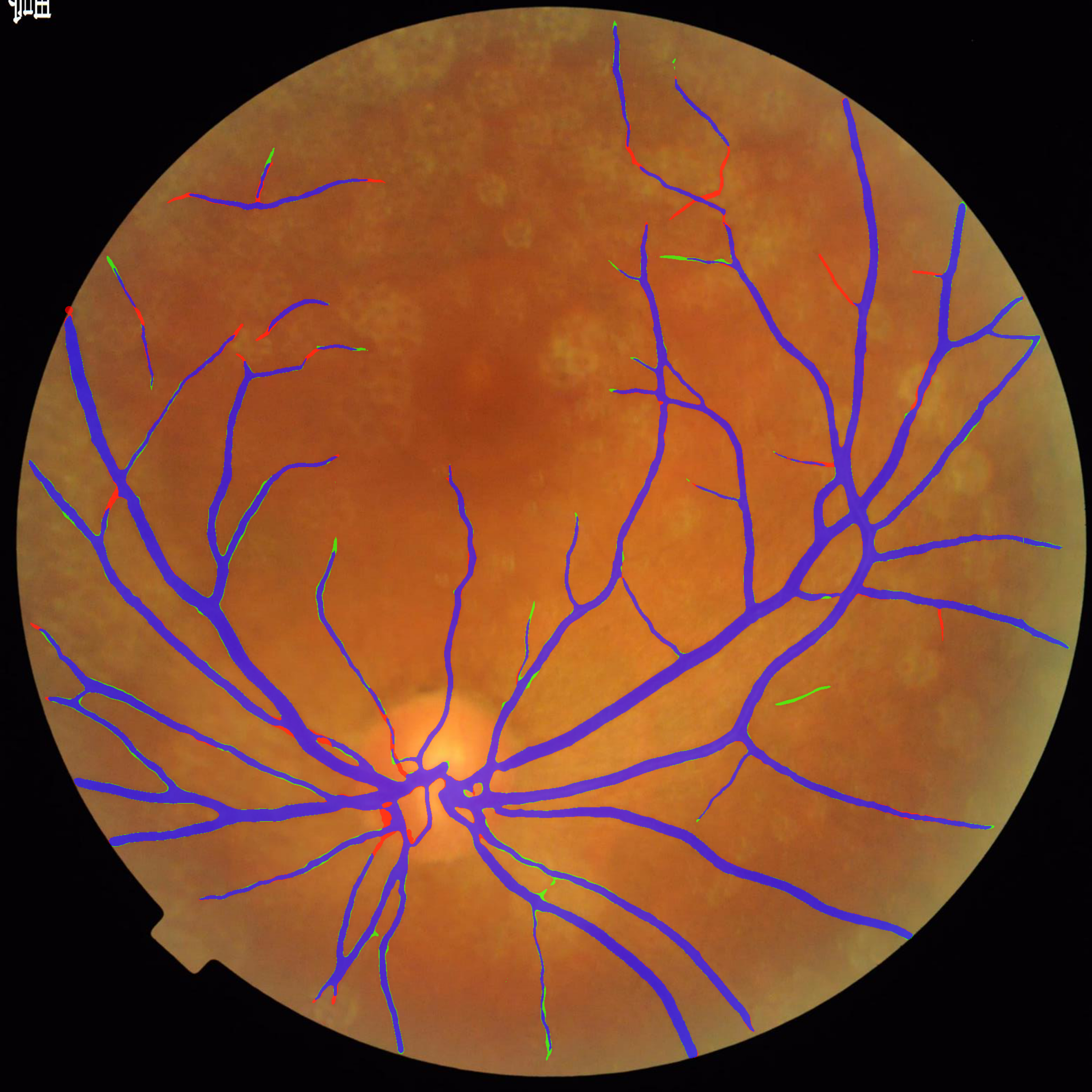}{95.1}
        \caption*{SegResNet}
    \end{subfigure}%
    \begin{subfigure}{0.19\textwidth}
        \includegraphix[.98\textwidth]{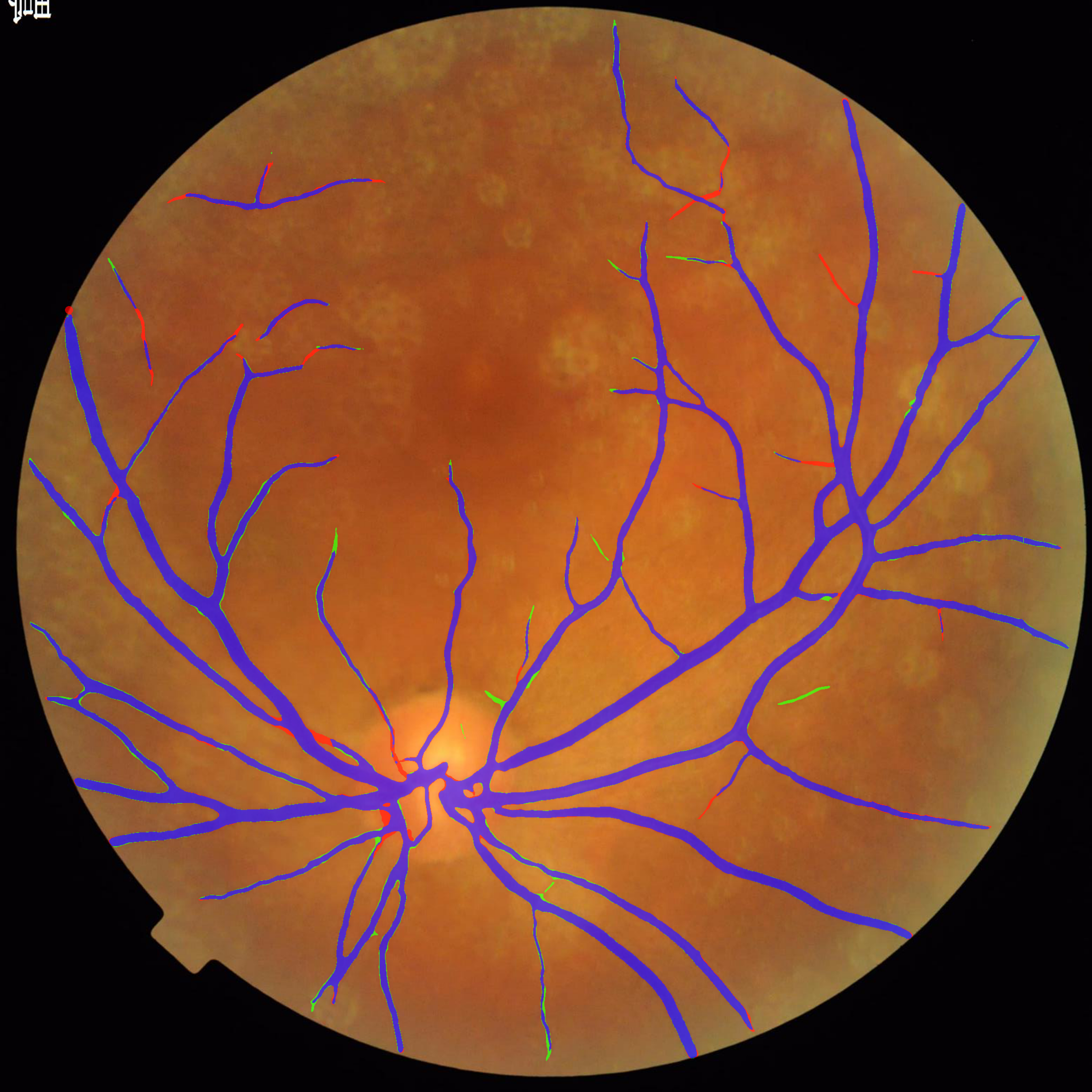}{95.2}
        \caption*{Swin UNETR}
    \end{subfigure}%
    \begin{subfigure}{0.19\textwidth}
        \includegraphix[.98\textwidth]{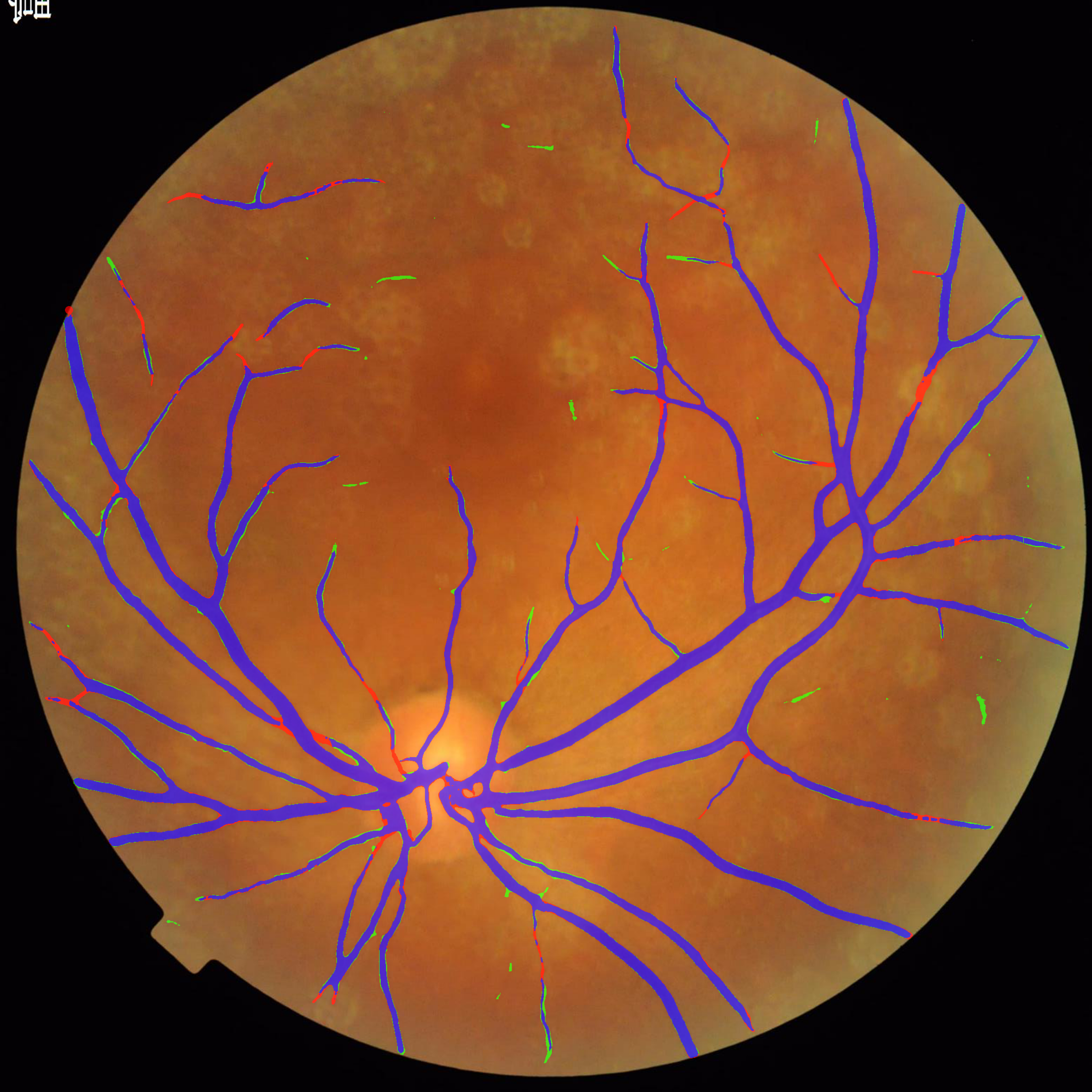}{93.2}
        \caption*{UNETR}
    \end{subfigure}

    \caption{Qualitative results of 2D segmentation on GlaS (first row) and FIVES (second row). The regions of true positives are marked in blue, false positives in green, and false negatives in red. The DSC score is presented for each case. The first row presents an example from the validation set of the GlaS dataset, where all baseline models undersegment the gland. The second row shows an example from the FIVES validation set, where, consistent with the quantitative results, most models perform similarly.
    }
    \label{fig:qualitative_results_2d}
\end{figure*}
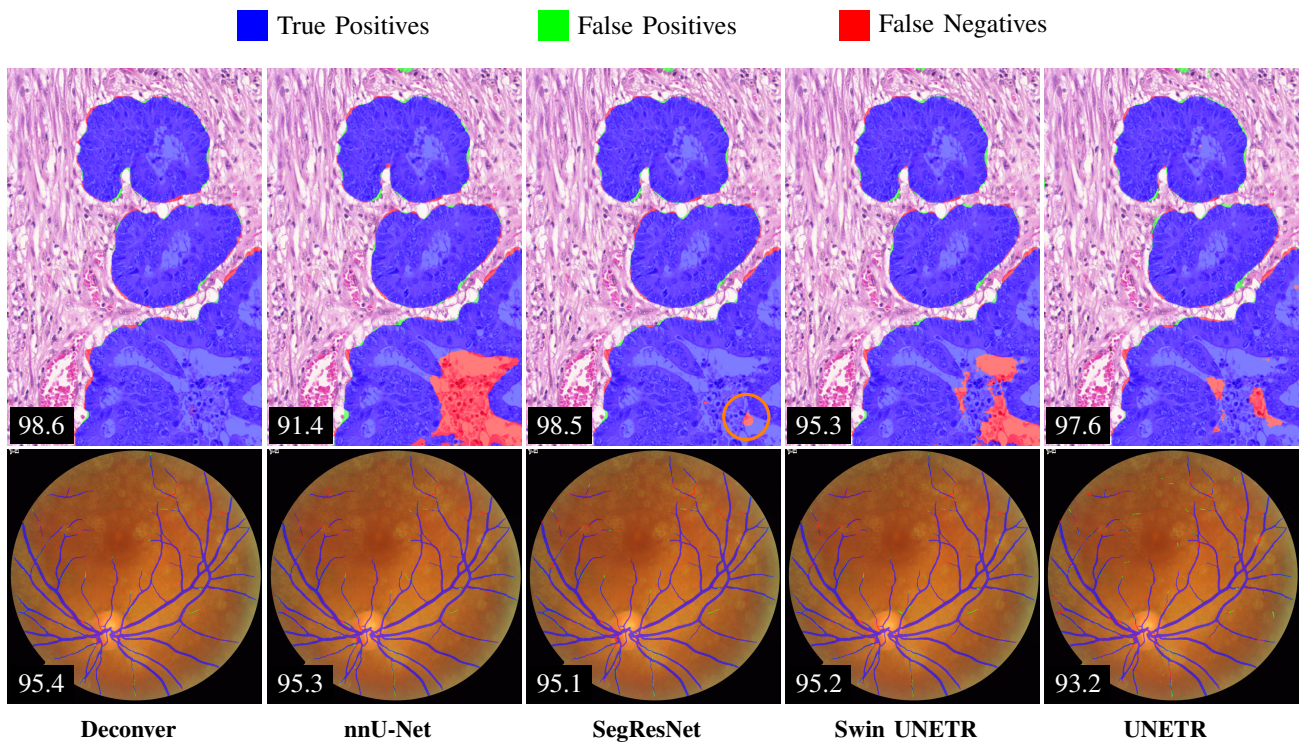

Consistent with results on 3D datasets, Deconver demonstrates an excellent trade-off between accuracy and computational efficiency on the GlaS and FIVES datasets. Both Deconver variants (\( 3 \times 3 \) and \( 5 \times 5 \)) achieve superior or competitive DSC scores compared to SegResNet, while requiring over 75\% fewer FLOPs per pixel. Furthermore, the parameter count for Deconver remains close to that of nnU-Net and SegResNet, and is nearly six times lower than UNETR.

In line with our analysis of 3D datasets, we also conducted a qualitative evaluation of the 2D segmentation results, as shown in Fig. \ref{fig:qualitative_results_2d}. The first row illustrates an example from the GlaS dataset, highlighting clear distinctions among the models. All baseline methods, notably nnU-Net and Swin UNETR, fail to completely segment the lowest gland, resulting in substantial false negatives, particularly in its central region. In contrast, Deconver accurately captures the full glandular structure, achieving an almost perfect segmentation mask. We used Deconver with the kernel size of 5 for this experiment.

The second row presents segmentation results from the FIVES dataset. Here, as anticipated by quantitative metrics, most models produce consistently accurate segmentations, with minimal observable differences between their outputs. Just one notable observation is that UNETR produces a higher number of false positives compared to the other methods.

\subsection{Ablation Studies} \label{sec:ablation_studies}

To better understand the impact of parameter choices on Deconver, we performed two ablation studies on the ISLES'22 dataset. We evaluated two key hyperparameters of NDC layers: the source channel ratio (\( R \)) and the number of groups (\( G \)). In both experiments, we varied one parameter while keeping the other fixed to assess its independent effect. In both ablation experiments, we used the kernel size of \((3, 3)\), which was found to be optimal for the ISLES'22 dataset.

\subsubsection{Effect of Channel Ratio (\( R \))}

In this experiment, we fixed \( G \) to the number of channels and varied \( R \) to analyze its effect on the performance (Table \ref{tab:Ablation_ratio}). We observe that increasing \( R \) from 1 to 2 does not change the results significantly, while increasing it from 1 to 4 leads to major improvements in DSC and HD95. However, these gains come at the cost of increased computational complexity as both the number of parameters and FLOPs increase by 35.23\% and 42.55\% respectively.

\begin{table}[!t]
    \centering
    \caption{The results of the ablation study on the ratio (R) parameter of Deconver using ISLES'22 dataset. Best results are \textbf{bold}, second-best are \underline{underlined}.}
    \label{tab:Ablation_ratio}
    \renewcommand{\arraystretch}{1.2}
    \resizebox{\linewidth}{!}{
        \begin{tabular}{l | c c c c}
            \toprule
            \textbf{Ratio (R)} & \textbf{Params} & \textbf{FLOPs / pixel } & \textbf{DSC (\%)} & \textbf{HD95}    \\
            \midrule
            1.0                & 7.8M            & 425.8K                  & \underline{77.39} & \underline{5.07} \\
            2.0                & 8.7M            & 486.2K                  & {77.32}           & {5.17}           \\
            4.0                & 10.5M           & 607.0K                  & \textbf{78.16}    & \textbf{4.99}    \\
            \bottomrule
        \end{tabular}
    }
\end{table}

\subsubsection{Effect of Number of Groups (\( G \))}

In this experiment, we fixed \( R = 4 \) and varied \( G \) to assess its impact (Table \ref{tab:Ablation_group}). The results reveal a surprising pattern. Decreasing \( G \) leads to a significant increase in parameters without proportional improvements in DSC. In particular, setting \( G = 1 \) drastically inflates the size of the model (almost 5 times more parameters) while providing no gains, setting \( G \) to the number of channels achieves the best DSC and HD95 while keeping the model compact. Notably, changing the Groups parameter does not lead to major changes in the FLOPs.

\begin{table}[!t]
    \centering
    \caption{The results of the ablation study on the Groups (G) parameter of deconver using ISLES'22 dataset. Best results are \textbf{bold}, second-best are \underline{underlined}.}
    \label{tab:Ablation_group}
    \renewcommand{\arraystretch}{1.2}
    \resizebox{\linewidth}{!}{
        \begin{tabular}{l | c c c c}
            \toprule
            \textbf{Groups} & \textbf{Params } & \textbf{FLOPs / pixel } & \textbf{DSC (\%)} & \textbf{HD95}    \\
            \midrule
            1               & 57.21M           & 607.4K                  & {77.76}           & {5.08}           \\
            8               & 16.18M           & 607.1K                  & \underline{77.90} & \underline{5.05} \\
            Channels        & 10.48M           & 607.0K                  & \textbf{78.16}    & \textbf{4.99}    \\
            \bottomrule
        \end{tabular}
    }
\end{table}

In general, these studies show that higher channel ratios improve segmentation quality but increase computational cost, while a higher number of groups significantly reduces the number of parameters while improving the performance. The best results are achieved when \( G \) is set to the number of channels and \( R = 4 \), as this configuration produces optimal segmentation accuracy with minimal overhead.

\section{Conclusion} \label{sec:conclusion}

In this work, we introduce Deconver, a powerful segmentation network that integrates nonnegative deconvolution (NDC) as a learnable module within a U-shaped architecture. By replacing computationally expensive attention mechanisms with efficient deconvolution operations, Deconver restores high-frequency details while effectively suppressing artifacts.

Extensive experiments on four diverse medical imaging datasets (ISLES'22, BraTS'23, GlaS, and FIVES) demonstrate that Deconver consistently achieves state-of-the-art segmentation performance, outperforming or matching leading CNN- and Transformer-based models while significantly reducing computational costs. Notably, Deconver reduces FLOPs by up to 90\% compared to attention-based baselines, making it well-suited for resource-constrained clinical applications. Ablation studies further highlight the importance of key design choices, such as the source channel ratio and grouping strategy in NDC layers, in balancing accuracy and efficiency. We believe Deconver represents a promising step toward high-precision, computationally efficient medical image segmentation, bridging the gap between classical image restoration and modern deep learning.

\appendix
\section*{Proof of Theorem \ref{the:monotonicity}} \label{app:monotonicity_proof}

We prove the theorem using the majorization-minimization (MM) framework. This involves iteratively minimizing a surrogate function that upperbounds the original objective. The MM approach guarantees a monotonic decrease in the reconstruction error \( \mathcal{E}(\tns{S}) \) through two key steps:
\begin{enumerate}
    \item \textbf{Majorization}: Construct a surrogate function \( Q(\tns{S} \mid \tns{S}^{(t)}) \) that satisfies:
          \begin{equation} \label{eq:majorization}
              Q(\tns{S} \mid \tns{S}^{(t)}) \geq \mathcal{E}(\tns{S}),
          \end{equation}
          for all \( \tns{S} \), with equality when \( \tns{S} = \tns{S}^{(t)} \).
    \item \textbf{Minimization}: Update the source to minimize the surrogate:
          \begin{equation} \label{eq:minimization}
              \tns{S}^{(t+1)} = \argmin_{\tns{S} \geq 0} Q(\tns{S} \mid \tns{S}^{(t)})
          \end{equation}
\end{enumerate}
Combining these, we directly obtain:
\begin{equation*} \label{eq:monotonicity}
    \mathcal{E}(\tns{S}^{(t+1)}) \leq Q(\tns{S}^{(t+1)} \mid \tns{S}^{(t)}) \leq Q(\tns{S}^{(t)} \mid \tns{S}^{(t)}) = \mathcal{E}(\tns{S}^{(t)}),
\end{equation*}
proving the reconstruction error is non-increasing across iterations.

\subsection*{Step 1: Majorization}
Let's first expand the reconstruction error as
\begin{align}
    \mathcal{E}(\tns{S}) & = \|\tns{X} - \tns{S} \ast \tns{V}\|_\text{F}^2 \nonumber                                                                                    \\
               & = \|\tns{X}\|_\text{F}^2 - 2 \langle \tns{X}, \tns{S} \ast \tns{V} \rangle + \|\tns{S} \ast \tns{V}\|_\text{F}^2. \label{eq:error_expansion}
\end{align}

The main challenge lies in majorizing the quadratic term \( \|\tns{S} \ast \tns{V}\|_F^2 \). To achieve this, we apply the elementwise Cauchy--Schwarz inequality to \( (\tns{S} \ast \tns{V})^2 \). For each output element \((c,h,w)\), define:
\begin{align*}
    A_{d,m,n} & = \frac{\tns{S}_p[d, h+m, w+n] \tns{V}[c,d,m,n]}{\sqrt{\tns{S}_p^{(t)}[d, h+m, w+n] \tns{V}[c,d,m,n]}}, \\
    B_{d,m,n} & = \sqrt{\tns{S}_p^{(t)}[d, h+m, w+n]  \tns{V}[c,d,m,n]},
\end{align*}
where \( \tns{S}_p^{(t)} = \text{pad}(\tns{S}^{(t)}, (M,N)) \) and \( \tns{S}_p = \text{pad}(\tns{S}, (M,N)) \). By Cauchy--Schwarz, we have:
\begin{equation*} \label{eq:cauchy-schwarz}
    \bigg( \sum_{d,m,n} A_{d,m,n} B_{d,m,n} \bigg)^2 \leq \bigg(  \sum_{d,m,n} A_{d,m,n}^2 \bigg) \bigg(  \sum_{d,m,n} B_{d,m,n}^2 \bigg).
\end{equation*}
Substituting back and using the definition of cross-correlation, we obtain:
\begin{equation}  \label{eq:quadratic_elementwise_upperbound}
    (\tns{S} \ast \tns{V})^2 \leq \left(\frac{\tns{S}^2}{\tns{S}^{(t)}} \ast \tns{V} \right) \odot \left( \tns{S}^{(t)} \ast \tns{V} \right).
\end{equation}
Where \( (\cdot)^2 \) denotes elementwise squaring. Summing over all elements gives:
\begin{align} \label{eq:quadratic_upperbound}
    \|\tns{S} \ast \tns{V}\|_\text{F}^2 & \leq \sum_{c,h,w} \left(\frac{\tns{S}^2}{\tns{S}^{(t)}} \ast \tns{V}\right)[c,h,w] \cdot \left(\tns{S}^{(t)} \ast \tns{V}\right)[c,h,w] \nonumber \\
                                        & = \langle \frac{\tns{S}^2}{\tns{S}^{(t)}} \ast \tns{V}, \tns{S}^{(t)} \ast \tns{V} \rangle.
\end{align}

This constructs the majorizing surrogate function:
\begin{equation} \label{eq:surrogate_function}
    Q(\tns{S} \mid \tns{S}^{(t)}) = \|\tns{X}\|_\text{F}^2 - 2 \langle \tns{X}, \tns{S} \ast \tns{V} \rangle + \langle \frac{\tns{S}^2}{\tns{S}^{(t)}} \ast \tns{V}, \tns{S}^{(t)} \ast \tns{V} \rangle.
\end{equation}

\subsection*{Step 2: Minimization}
To minimize the surrogate function \( Q(\tns{S} \mid \tns{S}^{(t)}) \), we first derive its gradient with respect to \( \tns{S} \).

\subsubsection*{Linear Term Gradient}
The linear term \( \langle \tns{X}, \tns{S} \ast \tns{V} \rangle \) has gradient:
\begin{equation} \label{eq:grad_linear}
    \nabla_\mathcal{S} \langle \tns{X}, \tns{S} \ast \tns{V} \rangle[d',h',w'] = \sum_{c,h,w} \tns{X}[c,h,w] \frac{\partial (\tns{S} \ast \tns{V})[c,h,w]}{\partial \tns{S}[d',h',w']}.
\end{equation}
Expanding the cross-correlation \( (\tns{S} \ast \tns{V})[c,h,w] \), we find that the partial derivative is nonzero when:
\begin{equation} \label{eq:index_condition}
    h + m = h' + M, \quad w + n = w' + N.
\end{equation}
Thus,
\begin{align}
    \frac{\partial (\tns{S} \ast \tns{V})[c,h,w]}{\partial \tns{S}[d',h',w']} = \begin{cases}
                                                                                    \tns{V}[c,d',m,n], & \text{if \eqref{eq:index_condition} holds}, \\
                                                                                    0,                 & \text{otherwise}.
                                                                                \end{cases}
    \label{eq:partial_derivative}
\end{align}
Substituting back, \( \nabla_\mathcal{S} \langle \tns{X}, \tns{S} \ast \tns{V} \rangle[d',h',w'] \) simplifies to:
\begin{equation*}  \label{eq:grad_linear_expanded}
    \sum_{c,m,n} \tns{X}[c, h' - m + M, w' - n + N] \tns{V}[c,d',m,n].
\end{equation*}
Recognizing this as a cross-correlation operation, we obtain:
\begin{align}
    \nabla_\mathcal{S} \langle \tns{X}, \tns{S} \ast \tns{V} \rangle & = \tns{X} \ast \tns{V}^-.
    \label{eq:grad_linear_final}
\end{align}

\subsubsection*{Quadratic Term Gradient}
Using \eqref{eq:grad_linear_final} together with chain rule, the gradient of the quadratic term can be derived as
\begin{align}
    \nabla_\mathcal{S} \langle \frac{\tns{S}^2}{\tns{S}^{(t)}} \ast \tns{V}, \tns{S}^{(t)} \ast \tns{V} \rangle & = \frac{2\tns{S}}{\tns{S}^{(t)}} \odot \left[ (\tns{S}^{(t)} \ast \tns{V}) \ast \tns{V}^- \right].
    \label{eq:grad_quad_final}
\end{align}

\subsubsection*{Solving for \( \tns{S}^{(t+1)} \)}
Combining both gradients and setting the total gradient \( \nabla_\mathcal{S} Q(\tns{S} \mid \tns{S}^{(t)}) = 0 \) yields
\begin{equation} \label{eq:total_grad}
    -2(\tns{X} \ast \tns{V}^-) + \frac{2\tns{S}}{\tns{S}^{(t)}} \odot \left[ (\tns{S}^{(t)} \ast \tns{V}) \ast \tns{V}^- \right] = 0.
\end{equation}
Solving for \( \tns{S} \), we derive the multiplicative update rule:
\begin{equation*} \label{eq:update_rule}
    \tns{S}^{(t+1)} = \tns{S}^{(t)} \odot \frac{\tns{X} \ast \tns{V}^-}{(\tns{S}^{(t)} \ast \tns{V}) \ast \tns{V}^-}.
\end{equation*}

If for any index \((d,h,w)\), we have \((\tns{S}^{(t)} \ast \tns{V} \ast \tns{V}^-)[d,h,w] = 0\), then the nonnegativity of \( \tns{X} \), \( \tns{V} \), and \( \tns{S}^{(t)} \) implies \((\tns{X} \ast \tns{V}^-)[d,h,w] = 0\). In this case, the indeterminate form \(0/0\) is resolved by setting \( \tns{S}^{(t+1)}[d,h,w] = 0 \), thereby preserving nonnegativity.

\printbibliography

\end{document}